\documentclass{article}

\usepackage{booktabs}
\usepackage{tabularx}
\usepackage{multirow}
\usepackage[table]{xcolor}
\usepackage{array}
\usepackage{booktabs}
\usepackage{multirow}
\usepackage{graphicx}
\usepackage{caption}
\usepackage{float}

\usepackage{microtype}
\usepackage{graphicx}
\usepackage{subcaption}
\usepackage{booktabs} 

\usepackage{hyperref}

\usepackage[preprint]{2026}

\usepackage{amsmath}
\usepackage{amssymb}
\usepackage{mathtools}
\usepackage{amsthm}

\usepackage[capitalize,noabbrev]{cleveref}

\theoremstyle{plain}
\newtheorem{theorem}{Theorem}[section]
\newtheorem{proposition}[theorem]{Proposition}
\newtheorem{lemma}[theorem]{Lemma}
\newtheorem{corollary}[theorem]{Corollary}
\theoremstyle{definition}
\newtheorem{definition}[theorem]{Definition}
\newtheorem{assumption}[theorem]{Assumption}
\theoremstyle{remark}
\newtheorem{remark}[theorem]{Remark}

\theoremstyle{definition}
\newtheorem{example}{Example}[section]
\usepackage[textsize=tiny]{todonotes}

\aatitlerunning{Feasible Fusion: Constrained Joint Estimation under Structural Non-Overlap}

\begin{document}

\twocolumn[
  \aatitle{Feasible Fusion: Constrained Joint Estimation under Structural Non-Overlap}

  \aasetsymbol{equal}{*}

  \begin{aaauthorlist}
    \aaauthor{Yuxi Du}{sufe}
    \aaauthor{Zhiheng Zhang}{sufe}
    \aaauthor{Haoxuan Li}{pku}
    \aaauthor{Cong Fang}{pku}
    \aaauthor{Jixing Xu}{comp}
    \aaauthor{Peng Zhen}{comp}
    \aaauthor{Jiecheng Guo}{comp}

  \end{aaauthorlist}

  \aaaffiliation{sufe}{School of Statistics and Data Science,Shanghai University of Finance and Economics,Shanghai,China}
  \aaaffiliation{comp}{Didi Chuxing,Beijing, China}
  \aaaffiliation{pku}{Peking University, Beijing, China}

  \aacorrespondingauthor{Zhiheng Zhang}{studyzzh@163.com}

  \aakeywords{Machine Learning}

  \vskip 0.3in
]

\printAffiliationsAndNotice{}

\begin{abstract}

Causal inference in modern large-scale systems faces growing challenges, including high-dimensional covariates, multi-valued treatments, massive observational (OBS) data, and limited randomized controlled trial (RCT) samples due to cost constraints. We formalize treatment-induced structural non-overlap and show that, under this regime, commonly used weighted fusion methods provably fail to satisfy randomized identifying restrictions.To address this issue,we propose a \emph{constrained joint estimation framework} that minimizes observational risk while enforcing causal validity through orthogonal experimental moment conditions. We further show that structural non-overlap creates a feasibility obstruction for moment enforcement in the original covariate space.We also derive a penalized primal–dual algorithm that jointly learns representations and predictors, and establish oracle inequalities decomposing error into overlap recovery, moment violation, and statistical terms.Extensive synthetic experiments demonstrate robust performance under varying degrees of non-overlap. A large-scale ride-hailing application shows that our method achieves substantial gains over existing baselines, matching the performance of models trained with significantly more RCT data.
\end{abstract}
\section{Introduction}
\label{sec:introduction}

Modern empirical decision-making increasingly relies on two complementary sources of evidence: randomized controlled trials (RCTs), which provide high internal validity through randomized treatment assignment, and large-scale observational (OBS) logs, which capture rich heterogeneity at scale but may suffer from confounding and policy-induced selection effects.
A central goal in contemporary causal inference is to combine these data sources so as to obtain estimators that are simultaneously \emph{causally valid} and \emph{statistically efficient}; see, e.g., \citet{Imbens_Rubin_2015}, \citet{degtiar2023review}, and the recent review by \citet{colnet2024combining}.
This goal is especially consequential in domains where small changes in decision rules can translate into large welfare or revenue impacts, such as online platforms, transportation marketplaces, and other large-scale allocation systems.


Within this broad agenda, we focus on \emph{heterogeneous response estimation} under \emph{structured treatments}.
Rather than a binary intervention, many real systems deploy multi-valued, high-cardinality, or combinatorial treatment assignments (e.g., discretized policy parameters, bundles of operational actions, or constrained action sets), for which classical propensity-score ideas must be interpreted in a generalized form \citep{imbens2000dose, imai2004general}.
In such settings, the target object is naturally expressed as a family of conditional response functions
$
\mu_t(x) = \mathbb{E}\!\left[ Y(t)\mid X=x \right],~ t\in\mathcal{T},$ where $X$ denotes covariates, $t$ denotes a structured treatment, and $Y(t)$ is the potential outcome under $t$.
A growing body of work studies how to leverage RCT information to ``ground'' or regularize estimates learned from observational logs (and conversely, how to improve the external validity of RCT findings using observational samples), under assumptions that make the two sources compatible \citep{degtiar2023review, colnet2024combining, bareinboim2016datafusion}.
At a methodological level, this literature connects to (i) generalizability/transportability analyses, (ii) debiasing or robustness-based fusion, and (iii) representation-learning approaches that attempt to mitigate distribution shift by learning balanced covariate embeddings \citep{johansson2016learning, shalit2017estimating, shi2019adapting}.

Despite this progress, much of the existing theory and practice of data fusion relies—often implicitly—on some form of overlap or positivity: treatments of interest are assumed to occur with non-negligible probability across relevant covariate strata \citep{rosenbaum1983central}.
When overlap is weak but nonzero, standard remedies include trimming regions with extreme propensities, reweighting samples, or tolerating extrapolation, each of which trades bias against variance and weakens finite-sample guarantees \citep{crump2009limited}. However, in many modern platform environments with structured and constrained decision rules, the dominant difficulty is of a qualitatively different nature.
Here, certain covariate--treatment combinations are not merely rare, but \emph{structurally impossible} under the production policy and therefore never appear in observational logs.
This absence is induced by the decision mechanism itself—such as hard constraints, safety rules, or operational policies—rather than by sampling variability or limited data.
As a consequence, increasing the size of observational data or adjusting weights cannot recover these missing regions of support, rendering classical overlap-based remedies ineffective.

Consider a ride-hailing platform (e.g., Didi) that operates under supply constraints, operational rules, and safety/business restrictions.
Treatments may encode combinations of dispatch actions, pricing multipliers, or policy parameters, and the deployed policy may deterministically exclude large subsets of the action set in particular contexts.
Consequently, there can exist measurable regions of $(x,t)$ for which an experimental design assigns positive probability (because exploration is explicitly introduced), while the observational logging policy assigns probability zero.
In such a regime, naive pooling---or any procedure that implicitly treats mixed RCT--OBS data as if generated from a single compatible design---faces a fundamental dilemma:
either it \emph{drifts} toward observationally optimal predictors that violate experimental identifying restrictions, or it \emph{extrapolates} individual-level effects into regions unsupported by the observational policy. 
Crucially, this failure mode is shared across a broad family of approaches, including weighted-loss fusion and representation-learning pipelines that are trained without explicitly encoding experimental identification, because the missing support cannot be created by tuning weights or by balancing objectives alone \citep{johansson2016learning, shalit2017estimating}.


Our contributions could be summarized as follows. 
\begin{itemize}
  \item We formalize \emph{treatment-induced structural non-overlap} for structured treatments and show that it creates a fundamental support and feasibility obstruction for standard overlap-based and weighted fusion of randomized and observational data.
  \item We propose a constrained joint estimation framework that enforces randomization-induced experimental moment conditions while minimizing observational risk, together with a representation-learning scheme that restores approximate feasibility via distribution alignment.
  \item We establish theoretical separation results and oracle-type risk bounds under structural non-overlap, and demonstrate substantial empirical gains in synthetic studies and a large-scale ride-hailing application where observational data are abundant but experiments are costly.
\end{itemize}

\textbf{Organization.}
Section~\ref{sec:setup} introduces the setup and the notion of treatment-induced structural non-overlap.
Section~\ref{sec:method} develops the constrained joint-estimation framework and presents the main
theoretical results on feasibility recovery and risk bounds.
Section~\ref{sec:syn_test} and \ref{sec:exp} reports synthetic and real-data experiments

In contrast to prior work, we adopt a constrained joint estimation paradigm that explicitly separates causal validity from observational efficiency.
Rather than soft fusion via weighted empirical risks, we treat experimental randomization as a source of moment constraints that must be satisfied exactly in the population.
Moreover, unlike two-stage pipelines that first learn representations and subsequently enforce causal restrictions, our framework learns representations jointly with the experimental constraints, recognizing feasibility restoration as an inherently causal—rather than purely predictive—objective.
This joint, constraint-driven view enables identification and estimation in regimes where existing integrative methods provably break down.

\section{Problem Setup}
\label{sec:setup}

We consider units indexed by $i=1,\dots,n$ with covariates $X\in\mathcal X\subset\mathbb R^d$,
structured treatments $T\in\mathcal T$, and outcomes $Y\in\mathbb R$.
The treatment space $\mathcal T$ may be multi-valued or combinatorial, representing policy rules
or action bundles rather than a binary intervention.
Potential outcomes $\{Y(t):t\in\mathcal T\}$ are defined in the Rubin causal model.

We observe two data sources generated under distinct assignment mechanisms:
\textbf{(i) RCT data} from $P_r(X,T,Y)$, where treatment is randomized conditional on $X$,
and \textbf{(ii) observational (OBS) data} from $P_o(X,T,Y)$, where treatment follows a production
policy that may depend on $X$ and unobserved factors.
The distributions $P_r$ and $P_o$ may differ in both covariate marginals and treatment assignment.

Our target estimand is the conditional response function
\begin{equation}
\mu_t(x) = \mathbb E[Y(t)\mid X=x], \qquad t\in\mathcal T,
\end{equation}
which captures heterogeneous, individual-level causal effects under structured treatments.

Classical identification assumes overlap,
$P_o(T=t\mid X=x)>0$ for all $(x,t)$.
In contrast, under constrained or rule-based assignment, overlap violations are \emph{structural}:
certain covariate--treatment pairs are never generated by the observational policy, irrespective of
sample size.

Next, we will provide two specific definitions of structural non overlap and explain why Definition~\ref{def:marginal_nonoverlap} can only be solved by extrapolation, inevitably introducing model spcification bias, while Definition~\ref{def:conditional_nonoverlap} can recover representations through joint training, thus becoming the research topic of this article.

\begin{definition}[Marginal (irreducible) structural non-overlap]
\label{def:marginal_nonoverlap}
We say that the observational and randomized distributions exhibit
\emph{marginal structural non-overlap} if there exists a measurable set
$A\subseteq\mathcal T$ such that
\[
P_r(T\in A)>0
\qquad\text{and}\qquad
P_o(T\in A)=0.
\]
In this case, the support mismatch already occurs at the marginal level of the treatment,
independently of the covariates $X$.
\end{definition}

\begin{definition}[Conditional (recoverable) structural non-overlap]
\label{def:conditional_nonoverlap}
We say that the observational and randomized distributions exhibit
\emph{conditional (recoverable) structural non-overlap} if the treatment
has overlapping marginal support under $P_r$ and $P_o$, i.e.,
\[
\forall A \subseteq \mathcal T,\quad
P_o(T \in A) = 0 \;\Rightarrow\; P_r(T \in A) = 0,
\]
but there exists a measurable set $S \subseteq \mathcal X \times \mathcal T$ such that
\[
P_r(S) > 0
\quad \text{and} \quad
P_o(S) = 0.
\]
That is, the marginal support of the treatment is aligned under $P_r$ and $P_o$,
while the joint support mismatch arises conditionally on $X$.
\end{definition}

\begin{remark}[Recoverable versus irreducible non-overlap]
\label{rem:recoverable_extrapolation_small_sample}
Definitions~\ref{def:marginal_nonoverlap} and~\ref{def:conditional_nonoverlap} distinguish recoverable
and irreducible sources of structural non-overlap.
Under marginal structural non-overlap, the mismatch occurs in the marginal treatment support and is
irreducible: since representations only transform covariates, no $\phi$ can eliminate this mismatch
(Theorem~\ref{thm:irreducible_overlap}).
Extrapolation-based approaches implicitly operate in this regime by imposing strong outcome-model
assumptions to predict unsupported treatment effects; however, such methods do not restore
identification via randomized moment validity and may incur persistent misspecification bias.
\end{remark}

In contrast to irreducible marginal non-overlap, conditional structural non-overlap admits that a representation $\phi(X)$ can
aggregate regions with zero observational support into equivalence classes with positive
support under the observational policy, thereby restoring overlap in the induced space.
This enables approximate randomized moment validity without extrapolation and clarifies
the distinction between recoverable and irreducible non-overlap underlying the feasibility
guarantee in Corollary~\ref{cor:strict_improvement} and the impossibility result in
Theorem~\ref{thm:irreducible_overlap}.
\paragraph{Representation learning for feasibility recovery.}
Motivated by this insight, we introduce a representation map
\begin{equation}
\phi:\mathcal{X}\to\mathcal{Z},
\end{equation}
whose role is to \emph{endogenously restore feasibility} of causal identification by balancing
overlap recovery against outcome-relevant information preservation, rather than performing
generic dimensionality reduction, invariance, or distributional balancing.

\begin{definition}[Representation Property for Joint Estimation]
\label{def:representation_property}
A representation map $\phi$ is said to satisfy the \emph{representation property for joint estimation} if it fulfills both the overlap recovery and information preservation conditions stated below.
\end{definition}

\begin{proposition}[Overlap Recovery]
\label{prop:overlap}
Let $\mathcal{D}$ be a class of measurable test functions on $\mathcal{Z}\times\mathcal{T}$ and let
$\mathrm{IPM}_{\mathcal{D}}(\cdot,\cdot)$ denote the associated integral probability metric.
The representation $\phi$ is said to recover overlap if there exists $\varepsilon_{\mathrm{ov}}>0$ such that
\begin{equation}
\mathrm{IPM}_{\mathcal{D}}\!\left(
P_r(\phi(X),T),\, P_o(\phi(X),T)
\right)
\;\le\;
\varepsilon_{\mathrm{ov}},
\label{eq:overlap_recovery}
\end{equation}
even though the original covariate space may exhibit violated overlap in the sense that
$\mathrm{supp}(X\mid T=t_1)\cap \mathrm{supp}(X\mid T=t_0)=\varnothing$
for some $t_0,t_1\in\mathcal{T}$.
\end{proposition}

\begin{proposition}[Outcome-Relevant Information Preservation]
\label{prop:information}
Let $\mathcal{F}$ be a hypothesis class of measurable functions on $\mathcal{X}$.
The representation $\phi$ is said to preserve outcome-relevant information if there exists
$\varepsilon_{\mathrm{info}}>0$ such that
\begin{equation}
\inf_{f\in\mathcal{F}}
\mathbb{E}\!\left[(Y-f(\phi(X)))^2\right]
\;\le\;
\inf_{g\in\mathcal{F}}
\mathbb{E}\!\left[(Y-g(X))^2\right]
\;+\;
\varepsilon_{\mathrm{info}},
\label{eq:information_preservation}
\end{equation}
that is, conditioning on $\phi(X)$ incurs at most an $\varepsilon_{\mathrm{info}}$ excess prediction risk
relative to conditioning on the full covariate vector $X$.
\end{proposition}

\paragraph{Technical challenge.}
Propositions~\ref{prop:overlap} and~\ref{prop:information} impose competing requirements on the
representation $\phi$.
Overlap recovery Eq~\eqref{eq:overlap_recovery} favors compressive mappings that reduce discrepancy between
$P_r$ and $P_o$, whereas information preservation Eq~\eqref{eq:information_preservation} requires retaining
outcome-relevant structure.
These objectives are inherently in tension(see Theorem~\ref{thm:feasibility_info_tradeoff}), especially in high-dimensional settings with structured
treatments, and cannot be simultaneously ensured by representations learned independently of causal
constraints.

\begin{remark}
\label{rem:challenge}
Most existing representation-learning approaches for causal inference implicitly assume that a representation satisfying both
Eq~\eqref{eq:overlap_recovery} and Eq~\eqref{eq:information_preservation} exists and can be obtained through balancing or invariance-based objectives.
Under treatment-induced structural non-overlap, this assumption is neither verifiable from data nor robust to model misspecification.
In contrast, our framework treats representation learning as an endogenous component of causal estimation,
explicitly quantifying feasibility through $\varepsilon_{\mathrm{ov}}$ and $\varepsilon_{\mathrm{info}}$ and analyzing how these quantities
propagate into estimation error in subsequent sections.
\end{remark}

\section{Method: From Oracle Moment-based Joint Estimation to the Representation Space}
\label{sec:method}

This section develops a constrained joint-estimation framework for fusing observational logs with randomized data in the presence of \emph{Conditional (recoverable) structural
non-overlap} (Definition~\ref{def:conditional_nonoverlap}).
We proceed in four steps.
Section~\ref{sec:method:moments} derives experimental moment conditions and explains why they constrain \emph{causal components} while leaving nuisance structure unconstrained.
Section~\ref{sec:method:constrained} formulates data fusion as a constrained program and establishes a separation from weighted-loss fusion.
Section~\ref{sec:method:feasibility} shows how structural non-overlap manifests as a \emph{feasibility obstruction} in the original covariate space and motivates representation learning as feasibility recovery.
Section~\ref{sec:method:algorithm} presents a primal--dual algorithm and a penalized objective suitable for large-scale training, and Section~\ref{sec:method:oracle} states an oracle inequality that decomposes error into overlap, moment, and statistical terms.

\subsection{Experimental moments as identifying restrictions}
\label{sec:method:moments}


For clarity of exposition, assume the treatment space is finite,
$\mathcal{T}=\{0,1,\dots,K\}$, where $0$ denotes a reference treatment (e.g., control) and $k\in\{1,\dots,K\}$ denote active treatments.
Write the one-hot encoding $D=(D_0,\dots,D_K)$ with $D_k=\mathbf{1}\{T=k\}$ and $\sum_{k=0}^K D_k=1$.
In an RCT, the assignment probabilities
$p_k \;:=\; \mathbb{P}_r(T=k)\qquad (k=0,\dots,K)$
are known by design (more generally, under stratified randomization one may take $p_k(x)=\mathbb{P}_r(T=k\mid X=x)$; all statements below extend by replacing $p_k$ with $p_k(X)$).



Two simple identities motivate our construction.
First, the observed outcome admits the potential-outcome decomposition
$Y = \sum_{k=0}^K D_k\,Y(k)
= Y(0) + \sum_{k=1}^K D_k\bigl(Y(k)-Y(0)\bigr),$ which is not an assumption, but a re-expression of $Y=\sum_k D_kY(k)$.
Second, any measurable regression function $m(x,t)$ can be written as
$m(x,t)= u(x) + \sum_{k=1}^K \mathbf{1}\{t=k\}\,h_k(x),
u(x):=m(x,0),
h_k(x):=m(x,k)-m(x,0).$
Thus, the pair $(u,\{h_k\}_{k=1}^K)$ separates a \emph{baseline} term $u$ (common across treatments) from \emph{treatment-specific} increments $h_k$.

Define, for each $k\in\{1,\dots,K\}$, the moment map
\begin{equation}
\psi_k(Y,T,X; m):= (D_k-p_k)\,\bigl(Y-m(X,T)\bigr),
\label{eq:moment_map}
\end{equation}
where $\psi := (\psi_1,\dots,\psi_K)^\top$. Under randomization, the conditional mean function
$m_r(x,t):=\mathbb{E}_{P_r}[Y\mid X=x,T=t]$ satisfies
\begin{equation}
\left\|\mathbb{E}_{P_r}\!\left[\psi(Y,T,X; m_r)\right] \;\right\|=\; 0.
\label{eq:moment_true}
\end{equation}

We treat \eqref{eq:moment_true} as an identifying restriction: any causally admissible predictor should satisfy the moment conditions (exactly in the population, approximately in finite samples).

\begin{lemma}[Baseline invariance of experimental moments]
\label{lem:baseline_invariance}
Let $u:\mathcal{X}\to\mathbb{R}$ be measurable.
Under RCT randomization with known assignment probabilities $p_k$,
$
\mathbb{E}_{P_r}\bigl[(D_k-p_k)\,u(X)\bigr] \;=\; 0,
\qquad \forall k\in\{1,\dots,K\}.
$
Consequently, for any $m$ parameterized as above,
the moment $\mathbb{E}_{P_r}[\psi_k(Y,T,X;m)]$ depends on $m$ only through the treatment-specific components $\{h_k\}$.
\end{lemma}

Expanding the moment using above equations yields
\begin{align}
&\mathbb{E}_{P_r}\!\left[(D_k-p_k)\,(Y-m(X,T))\right]\\
=&
\underbrace{\mathbb{E}_{P_r}\!\left[(D_k-p_k)\,(Y(0)-u(X))\right]}_{\text{baseline / nuisance component}}
\nonumber\\
&\quad+
\underbrace{\mathbb{E}_{P_r}\!\left[(D_k-p_k)\sum_{i=1}^K D_i\bigl(Y(i)-Y(0)-h_i(X)\bigr)\right]}_{\text{causal component}}.
\label{eq:moment_decomp}
\end{align}
By Lemma~\ref{lem:baseline_invariance}, the baseline term vanishes under randomization, leaving only the causal component.
This is the key reason to use moments rather than an RCT squared loss: the moment conditions constrain the \emph{causal increments} $\{h_k\}$ while remaining agnostic about baseline outcome structure. Importantly, Equation~\eqref{eq:moment_true} is \emph{not} equivalent to fitting $m$ by minimizing an RCT squared loss, which couples baseline and causal components and may enforce spurious agreement between RCT and OBS outcome models.
A common misconception is that ``giving the RCT a large weight'' in a weighted regression objective should recover causal validity.
Theorem~\ref{thm:weighted_separation} below formalizes why this intuition can fail under structural non-overlap: weighted-loss fusion may remain bounded away from satisfying \eqref{eq:moment_true}, regardless of tuning.(In practice, RCT samples are often scarce, so pushing $\alpha$ into the RCT-dominated regime can yield high-variance and unstable estimates.)

\subsection{Fusion as a constrained program: optimality and separation}
\label{sec:method:constrained}

Let $\mathcal{M}$ be a model class of measurable predictors $m:\mathcal{X}\times\mathcal{T}\to\mathbb{R}$. Define the observational risk $\mathcal{R}_o(m):=\mathbb{E}_{P_o}\bigl[(Y-m(X,T))^2\bigr]$. Guided by the complementary roles of the two data sources, we propose the constrained program
\begin{equation}
\min_{m\in\mathcal{M}}\;\mathcal{R}_o(m)
\qquad \text{s.t.}\qquad
\mathbb{E}_{P_r}\!\left[\psi(Y,T,X;m)\right]=0,
\label{eq:constrained_program}
\end{equation}
where $\psi$ is defined in \eqref{eq:moment_map}.
The objective exploits the abundance of observational data to learn high-dimensional nuisance structure, while the constraint enforces causal validity through experimental randomization.This distinction is particularly important in modern applications, where randomized traffic is scarce and the number of treatment arms is large. In such regimes, RCT-only estimation suffers from severe small-sample effects, leading to high-variance and unstable estimates of treatment effects, especially at the individual level. The constrained program mitigates this issue by decoupling nuisance learning from causal calibration: observational data are used to estimate complex baseline and representation components, while randomized data are used only to enforce causal validity.

\begin{assumption}[Constraint feasibility and regularity]
\label{assump:feas_reg}
The feasible set $\mathcal{F}
~:=~
\left\{m\in\mathcal{M}:\mathbb{E}_{P_r}\!\left[\psi(Y,T,X;m)\right]=0\right\} \neq \emptyset
$. Moreover, $\mathcal{M}$ is convex and closed, and the mapping
$m\mapsto \mathbb{E}_{P_r}[\psi(Y,T,X;m)]$ is continuous.
\end{assumption}

Under Assumption~\ref{assump:feas_reg}, any solution $m^\star$ to \eqref{eq:constrained_program} satisfies
$
\mathcal{R}_o(m^\star)=\inf_{m\in\mathcal{F}}\;\mathcal{R}_o(m).
$
Equivalently, \eqref{eq:constrained_program} returns the best observational predictor among all models that satisfy the randomization-induced moment restrictions.

A widely used alternative is to combine RCT and OBS objectives via a weighted squared loss:
\begin{equation}
m_\alpha \in \arg\min_{m\in\mathcal{M}}
\Big\{
\mathcal{R}_o(m)+\alpha\,\mathcal{R}_r(m)
\Big\},
\label{eq:weighted_fusion}
\end{equation}
where $\mathcal{R}_r(m):=\mathbb{E}_{P_r}\!\left[(Y-m(X,T))^2\right]$ with tuning parameter $\alpha\ge 0$.
Unlike \eqref{eq:constrained_program}, \eqref{eq:weighted_fusion} does \emph{not} encode experimental identification, and may therefore trade off causal validity against predictive fit.The remainder of Section~\ref{sec:method:constrained} shows that, in typical regimes, weighted-loss fusion cannot enforce the experimental moment conditions; moreover, aggressive tuning effectively reduces to an RCT-only estimator, yielding high-variance and unstable treatment-effect estimates.



\begin{assumption}[When weighted-loss fusion cannot satisfy experimental moments]
\label{assump:weighted_incompat}
\par
\noindent\textbf{(i) Observational fit gap induced by causal validity.}
There exists $\delta_o>0$ such that
\begin{equation}
\inf_{m\in\mathcal{F}} \mathcal{R}_o(m)
\;\ge\;
\inf_{m\in\mathcal{M}} \mathcal{R}_o(m)
\;+\;
\delta_o,
\label{eq:obs_gap}
\end{equation}
i.e., enforcing the randomized moment restrictions incurs a strictly positive increase in
observational prediction risk.

\noindent\textbf{(ii) Bounded range of the RCT risk.}
Define
\begin{equation}
B_r
~:=~
\sup_{m\in\mathcal{M}} \mathcal{R}_r(m)
-
\inf_{m\in\mathcal{M}} \mathcal{R}_r(m),
\qquad
0 < B_r < \infty.
\label{eq:rct_range}
\end{equation}

\noindent\textbf{(iii) Non-degenerate fusion regime and regularity.}
Consider weighted-loss fusion with $\alpha\in[0,\bar\alpha]$, where
\begin{equation}
\bar\alpha \;:=\; \frac{\delta_o}{2B_r}.
\label{eq:alpha_bar}
\end{equation}
For each $\alpha\in[0,\bar\alpha]$, the weighted objective~\eqref{eq:weighted_fusion} admits at least one
minimizer $m_\alpha\in\mathcal{M}$, and the set $\{m_\alpha:\alpha\in[0,\bar\alpha]\}$ is compact in
$L_2(P_r)$.
Moreover, the map $m\mapsto \mathbb{E}_{P_r}[\psi(Y,T,X;m)]$ is continuous on $\mathcal{M}$.
\end{assumption}


\begin{theorem}[Separation from weighted-loss fusion in the fusion regime]
\label{thm:weighted_separation}
Under Assumptions~\ref{assump:feas_reg} and \ref{assump:weighted_incompat},
the constrained solution $m^\star$ to \eqref{eq:constrained_program} satisfies the experimental moments exactly:
\[
\left\|
\mathbb{E}_{P_r}\!\left[\psi(Y,T,X;m^\star)\right]
\right\|=0.
\]
Moreover, there exists a constant $c_0>0$ such that every weighted-loss minimizer $m_\alpha$ of \eqref{eq:weighted_fusion}
with $\alpha\in[0,\bar\alpha]$ violates the moments by at least $c_0$:
\[
\inf_{\alpha\in[0,\bar\alpha]}
\left\|
\mathbb{E}_{P_r}\!\left[\psi(Y,T,X;m_\alpha)\right]
\right\|
\;\ge\;
c_0.
\]
Consequently, weighted-loss fusion with non-negligible observational weight
($\alpha \le \bar\alpha$) is fundamentally incompatible with exact satisfaction of the randomized moments,
which can only be achieved by taking $\alpha > \bar\alpha$.
\end{theorem}


Furthermore, Corollary~\ref{cor:g_linear} controls the weighted-solution path by showing that
$\|g(m_\alpha)\|$ can decrease at most linearly in $\alpha$.
This rules out an isolated ``sweet spot'' at small $\alpha$ where feasibility would suddenly hold
(see Appendix~\ref{sec:appendix:tradeoff_exclusion}).

We also prove that weighted-loss fusion coincides with the constrained estimator only in \emph{degenerate} settings(Assumption~\ref{assump:weighted_incompat} fails),requiring strong alignment between observational risk,
RCT risk, and the identifying moments (e.g., shared minimizers, exact KKT alignment).(see Appendix~\ref{sec:appendix:degeneracy_conditions}).
Such conditions are highly idealized and are generically violated under structural non-overlap, where
exact moment satisfaction incurs a positive observational risk penalty.
Consequently, weighted-loss fusion exhibits an irreducible moment violation in non-degenerate regimes.

This analysis above motivates the constrained formulation in Section~\ref{sec:method:feasibility}, which enforces causal validity
explicitly and uses observational data only within the causally admissible set, rather than relying on
scalar reweighting to implicitly satisfy identifying restrictions.


\subsection{Structural non-overlap as a feasibility obstruction and the role of representation}
\label{sec:method:feasibility}

The constrained formulation in Section~\ref{sec:method:constrained} assumes feasibility of the
identifying moment constraints within the function class $\mathcal M$
(Assumption~\ref{assump:feas_reg}).
Under conditional non-overlap (Definition~\ref{def:conditional_nonoverlap}), however, this assumption
may fail: no predictor in $\mathcal M$ may simultaneously satisfy the experimental moments and achieve
adequate observational fit in the original covariate space.

Rather than imposing feasibility \emph{a priori}, we view it as a property that may be
\emph{recovered or approximated} through representation learning.
Accordingly, this section characterizes when representation learning can restore feasibility and when
causal validity remains provably unattainable.

\begin{definition}[Feasibility gap]
\label{def:feas_gap}
Let $\psi$ be the moment map in \eqref{eq:moment_map}.
The \emph{feasibility gap} of $\mathcal{M}$ in the original covariate space is
$\mathcal F_{X}(\mathcal M)
:=
\inf_{m\in\mathcal{M}}
\left\|
\mathbb{E}_{P_r}\!\left[\psi(Y,T,X;m)\right]
\right\|.$
For a representation $\phi:\mathcal{X}\to\mathcal{Z}$ and induced class
$\mathcal{M}_\phi:=\{(x,t)\mapsto m(\phi(x),t):m\in\mathcal{M}\}$,
define the latent-space feasibility gap
$\mathcal F_{\phi}(\mathcal M)
:=
\inf_{m\in\mathcal{M}_\phi}
\left\|
\mathbb{E}_{P_r}\!\left[\psi(Y,T,\phi(X);m)\right]
\right\|.$
\end{definition}

A strictly positive $\mathcal F_{X}(\mathcal M)$ means that exact moment satisfaction is unattainable in the original space; feasibility must be recovered by altering the hypothesis class, which we do via representation learning.


\begin{theorem}[Explicit feasibility gap in the original space]
\label{thm:explicit_gap}
Let $\mathcal M$ be a model class of measurable predictors
$m:\mathcal X\times\mathcal T\to\mathbb R$ 
Under Assumption~\ref{def:conditional_nonoverlap} (so that there exists a measurable
$S\subseteq\mathcal X\times\mathcal T$ with $P_r(S)>0$ and $P_o(S)=0$),
Assumptions~\ref{assump:moment_reg}--\ref{assump:nondegenerate_assignment}, and
Lemma~\ref{lem:g_to_ro}, the feasibility gap in the original covariate space is strictly positive:
\[
\mathcal F_X(\mathcal M)
=
\inf_{m\in\mathcal M}
\Bigl\|
\mathbb E_{P_r}\!\big[\psi(Y,T,X;m)\big]
\Bigr\|
\;\ge\;
c_0,
\]
where
\[
c_0
~:=~
\underline p(1-\bar p)\,P_r(S)\,\delta
\;>\;0,
\]
and $(\underline p,\bar p)$ and $\delta$ are the constants in
Assumptions~\ref{assump:nondegenerate_assignment} and~\ref{assump:moment_misspec}, respectively.
\end{theorem}

Theorem~\ref{thm:explicit_gap} formalizes a feasibility \emph{obstruction} in the original covariate
space: under conditional structural non-overlap and mild misspecification, the identifying moments incur
a strictly positive residual $c_0$, i.e., exact feasibility is unattainable within $\mathcal M$ on $X$.
This suggests that enforcing the randomized moments necessarily requires altering the effective
hypothesis class.

Motivated by Theorem~\ref{thm:explicit_gap}, we pursue feasibility recovery by \emph{changing the space}
in which the identifying moments are enforced.
Rather than fixing a representation \emph{a priori}, we consider representations $\phi$ as part of the
optimization, each inducing a hypothesis class $\mathcal M_\phi$ with a potentially different
feasibility profile.
Our goal is therefore to identify representations under which the randomized moments become
(approximately) satisfiable, while simultaneously minimizing observational risk.

This perspective leads to the following latent constrained program:
\begin{equation}
\min_{\phi\in\Phi}\;\min_{m\in\mathcal M_\phi}\;\mathcal R_o(m)
\qquad \text{s.t.}\qquad
\mathbb{E}_{P_r}\!\left[\psi\!\left(Y,T,\phi(X);m\right)\right]=0,
\label{eq:phi_constrained_program}
\end{equation}
which should be interpreted as a \emph{joint} optimization over representations and predictors,
rather than a sequential two-stage procedure.
This formulation directly bridges the raw-space infeasibility established in
Theorem~\ref{thm:explicit_gap} with the strict feasibility improvement guaranteed next in
Corollary~\ref{cor:strict_improvement}.

\begin{corollary}[Strict improvement via representation learning]
\label{cor:strict_improvement}
Let $\mathcal M$ be a class of measurable predictors $m:\mathcal X\times\mathcal T\to\mathbb R$.
For any representation $\phi:\mathcal X\to\mathcal Z$, define the induced class
$\mathcal M_\phi:=\{m_\theta(\phi(\cdot),\cdot):\theta\in\Theta\}$ and the randomized moment map
\[
g_\phi(m)\;:=\;\mathbb E_{P_r}\!\big[\psi(Y,T,\phi(X);m)\big]\in\mathbb R^K
\]
Assume the raw-space obstruction $\mathcal F_X(\mathcal M)\ge c_0>0$, and that feasibility is
controlled by overlap in the induced space: there exists $c_1>0$ such that
$\mathcal F_\phi(\mathcal M)\le c_1\,\varepsilon_{ov}(\phi)$ for all $\phi\in\Phi$.
Moreover, whenever $\mathcal F_\phi\neq\emptyset$, Under Assumption~\ref{assump:eb_cr} and \ref{assump:g_lipschitz}: there exist $\mu_o>0$ and $L_g>0$ such that for all
$m\in\mathcal M_\phi$,
\[
\mathcal R_o(m)-\mathcal R_{o,\phi}^{\rm feas}\ \ge\ \mu_o\,\mathrm{dist}(m,\mathcal F_\phi)^2,
\qquad
\|g_\phi(m)\|\ \le\ L_g\,\mathrm{dist}(m,\mathcal F_\phi),
\]
where $\mathcal R_{o,\phi}^{\rm feas}:=\inf_{m\in\mathcal F_\phi}\mathcal R_o(m)$.

Let $\phi^\star$ be returned by the latent constrained program~\eqref{eq:phi_constrained_program},
and assume $\mathcal F_{\phi^\star}\neq\emptyset$. If $\varepsilon_{ov}(\phi^\star)<c_0/c_1$, then
\[
\mathcal F_{\phi^\star}(\mathcal M)\ \le\ c_1\varepsilon_{ov}(\phi^\star)\ <\ c_0\ \le\ \mathcal F_X(\mathcal M),
\]
so feasibility is strictly improved. In addition,
\begin{equation}
\inf_{m\in\mathcal M_{\phi^\star}}
\bigl\{\mathcal R_o(m)-\mathcal R_{o,\phi^\star}^{\rm feas}\bigr\}
\ \ge\
\frac{\mu_o}{L_g^2}\,\bigl(\mathcal F_{\phi^\star}(\mathcal M)\bigr)^2
\ \ge\
\frac{\mu_o}{L_g^2}\,\bigl(c_1\varepsilon_{ov}(\phi^\star)\bigr)^2.
\label{eq:best_excess_lb}
\end{equation}
\end{corollary}

We characterize the dependence of $c_1$ on the discriminator class $\mathcal D$ (via an embedding/dominance constant) in Appendix~\ref{app:explicit from of feasibility constant}.

Corollary~\ref{cor:strict_improvement} formalizes the central intuition behind representation learning
under \emph{conditional (recoverable) structural non-overlap} (Definition~\ref{def:conditional_nonoverlap}).
In this regime, although certain actions are never taken in specific contexts under the observational
policy (i.e., $P_o(x,t)=0$ while $P_r(x,t)>0$), the marginal support of treatments remains aligned between
observational and randomized data.
Consequently, there exists an appropriate representation $\phi^\star$ that coarsens or reparameterizes
the covariate space to recover joint support, rendering the identifying moments approximately feasible
and leading to strict improvements in both feasibility and observational risk.


However,even within the recoverable regime, feasibility recovery is not without cost.
Although representation learning avoids genuine extrapolation, the attainable improvement is jointly
constrained by residual distribution mismatch between observational and randomized data and by the
amount of outcome-relevant information retained.
Theorem~\ref{thm:feasibility_info_tradeoff} formalizes this limitation through an explicit
feasibility--information trade-off, showing that improving one without controlling the other is
insufficient.

This perspective also highlights a common pitfall: learning representations independently (e.g., for
prediction or balance) and enforcing causal restrictions afterward.
As feasibility restoration is inherently a causal objective, a representation may preserve predictive
information yet fail to reduce $\varepsilon_{\mathrm{ov}}$, leaving the feasible set
$\mathcal F_{\phi}(\mathcal M)$ large.
Our approach therefore learns $\phi$ jointly with the identifying moment constraints, ensuring that
overlap recovery and causal feasibility are optimized \emph{together}, as illustrated in the next
section.

Finally, feasibility recovery may be fundamentally impossible in certain regimes.
As shown by the minimax lower bound in Appendix~\ref{sec:appendix:degeneracy_conditions}, under
\emph{irreducible} structural non-overlap (Definition~\ref{def:marginal_nonoverlap}), no representation
can reduce the feasibility gap below a constant multiple of the residual overlap discrepancy.
Such regimes require genuine extrapolation beyond observational support and therefore fall outside the
scope of this work, which focuses on feasibility recovery via joint representation learning rather than
extrapolation.





\subsection{Augmented Lagrangian and a primal--dual algorithm}
\label{sec:method:algorithm}

In spired by Corollary~\ref{cor:strict_improvement} above,
we introduce augmented-Lagrangian
primal--dual scheme, which enforces the identifying moments while promoting overlap in the induced space in solving Eq~\eqref{eq:phi_constrained_program}.

Define the (population) randomized moment residual
\[
g(\theta,\phi)
:=\mathbb E_{P_r}\!\big[\psi(Y,T,\phi(X);m_\theta)\big]\in\mathbb R^K
\]
Let $R_o(\theta,\phi):=\mathbb E_{P_o}\!\big[(Y-m_\theta(\phi(X),T))^2\big]$ be the population observational risk.
We minimize the augmented Lagrangian
\begin{equation}
\begin{split}
\mathcal{L}_{\rho}(\theta, \phi, \nu)
&:= \mathcal{R}_o(\theta, \phi) + \lambda\,\varepsilon_{\rm ov}(\phi) \\
&\quad + \langle \nu,\, g(\theta, \phi) \rangle + \frac{\rho}{2}\|g(\theta, \phi)\|_2^2,
\end{split}
\label{eq:aug_lagrangian}
\end{equation}
where $\nu\in\mathbb{R}^K$ is the dual variable, $\rho>0$ is the augmentation parameter,
and $\lambda\ge 0$ controls overlap regularization.

We employ a stochastic primal--dual(PD) optimization scheme, alternating gradient descent on the primal variables $(\theta,\phi)$ and gradient ascent on the dual variable $\nu$ using minibatch estimates.
Algorithm~\ref{alg:pd_joint} provides the complete update procedure.Appendix~\ref{sec:appendix:pd_stability} proves PD is more stable than pure quadratic-penalty methods

\subsection{Oracle inequality and what it implies}
\label{sec:method:oracle}

We now connect algorithmic choice (provided in Section~\ref{sec:method:algorithm})to statistical performance:the following oracle inequality
quantifies how moment violation, residual non-overlap, finite-sample fluctuations, and optimization
suboptimality translate into excess observational risk.

\begin{theorem}[Oracle bound for a feasibility-augmented empirical program]
\label{thm:4.9-sharp}
Let $\widehat R_o(\theta,\phi)$ and $\widehat g(\theta,\phi)$ be empirical estimates based on
$n_o$ observational and $n_r$ randomized samples. Assume Assumptions~\ref{assump:overlap_to_feas}
and~\ref{assump:stat_error}. Define the feasibility-augmented objectives
\[
Q(\theta,\phi)
:= R_o(\theta,\phi)
+\frac{\mu_o}{L_g^2}\Big(\|g(\theta,\phi)\|+c_{\rm ov}\,\varepsilon_{\rm ov}(\phi)\Big)^2,
\]
\[
\widehat Q(\theta,\phi)
:= \widehat R_o(\theta,\phi)
+\frac{\mu_o}{L_g^2}\Big(\|\widehat g(\theta,\phi)\|+c_{\rm ov}\,\varepsilon_{\rm ov}(\phi)\Big)^2.
\]
Let $(\hat\theta,\hat\phi)$ be an $\eta$-approximate minimizer of $\widehat Q$, i.e., $\widehat Q(\hat\theta,\hat\phi)\le\inf_{\theta,\phi}\widehat Q(\theta,\phi)+\eta$.
Let $(\theta^\star,\phi^\star)$ be any oracle pair minimizing $Q$ (equivalently, any minimizer of
$\inf_{\theta,\phi} Q(\theta,\phi)$). Then with probability at least $1-\delta$,
\begin{equation}
\begin{split}
R_o(\hat{\theta},\hat{\phi})
&\leq R_o(\theta^\star,\phi^\star) + \frac{\mu_o}{L_g^2}\Bigl(\|g(\hat{\theta},\hat{\phi})\| + c_{\rm ov}\,\varepsilon_{\rm ov}(\hat{\phi})\Bigr)^2 \\
&\quad + C\,\mathrm{Stat}(n_o,n_r) + \eta,
\end{split}
\label{eq:oracle-quadratic}
\end{equation}
for a universal constant $C>0$.
\end{theorem}

Theorem~\ref{thm:4.9-sharp} establishes a near-oracle excess-risk bound for observational prediction.
Specifically,Eq.~\eqref{eq:oracle-quadratic} decomposes the excess risk into three terms: a quadratic
\emph{feasibility term}
$\frac{\mu_o}{L_g^2}\bigl(\|g(\hat\theta,\hat\phi)\| + c_{\rm ov}\,\varepsilon_{\rm ov}(\hat\phi)\bigr)^2$
capturing residual moment violation and overlap mismatch,a statistical error term
$C\cdot\mathrm{Stat}(n_o,n_r)$,and the optimization suboptimality $\eta$.

Crucially, the theorem explains why joint primal--dual training is necessary for achieving near-oracle risk.
Among these error sources, the feasibility residual $\|g(\hat\theta,\hat\phi)\|$ cannot, in general, be driven
to zero by methods that only reweight losses or fix the representation.
In contrast, a primal--dual (augmented Lagrangian) procedure directly enforces $g(\theta,\phi)=0$, and under
controlled overlap and statistical error, approximate saddle-point optimality implies
$\|g(\hat\theta,\hat\phi)\|\to 0$ (cf.\ Proposition~\ref{prop:joint-feasibility}).

\section{Synthetic Experiment}
\label{sec:syn_test}

This section presents two simulation studies.
\emph{(i)} We compare the proposed approach with baseline methods in terms of mean squared error (MSE) and Qini coefficient across datasets with systematically controlled levels of Conditional non-overlap(see Defination~\ref{def:conditional_nonoverlap}), assessing both predictive accuracy and treatment-effect ranking performance.
\emph{(ii)} We further investigate the robustness of the proposed method on multiple datasets exhibiting severe Conditional non-overlap, examining its stability and performance degradation relative to competing approaches.

\subsection{Baseline}
\label{sec:baselines}

We compare the proposed approach with a set of representative baseline methods(see Appendix~\ref{app:baseline_intro}) for integrating randomized controlled trial (RCT) and observational (OBS) data.
These methods reflect the evolution of data-fusion strategies, ranging from early structurally grounded approaches to modern, largely non-structural neural estimators.

\subsection{Data Generation Process And Joint Mismatch}

We construct a synthetic dataset that combines RCT and OBS samples to reflect realistic challenges arising from conditional structural non-overlap in high-dimensional covariate spaces.
The data generation process is described in Appendix~\ref{app:dgp}, and additional visualizations illustrating joint distribution mismatch are provided in Appendix~\ref{app:joint_mismatch}.

\subsection{Performance across datasets with varying overlap}
We construct three types of datasets exhibiting different levels of joint distribution mismatch.
Each dataset contains $n=5{,}000$ samples.
For each overlap setting, all methods are evaluated on 30 independently generated datasets, and we report the mean and standard deviation of the Qini coefficient and mean squared error (MSE) (see Table~\ref{tab:pc}).

Across all overlap regimes, the proposed method consistently attains the highest Qini scores with the
lowest variability, while also achieving near-minimal MSE in both mean and variance.
Notably, its performance remains stable as overlap deteriorates, whereas competing methods exhibit
substantial degradation when moving from large- to small-overlap settings.

Modern deep learning approaches may accommodate high-dimensional covariates (e.g., TARNet, DragonNet) or may incorporate distributional alignment objectives (e.g., CFRNet), these methods do not explicitly address conditional non-overlap, particularly in settings that require effective integration of randomized and observational data.

\subsection{Robustness under Severe Conditional Non-overlap}

To assess robustness in settings closer to real-world scenarios, we generate 50 datasets exhibiting varying degrees of severe conditional non-overlap.
In spite of different levels, all datasets in this setting suffer from \emph{Severe} conditional non-overlap.
Figure~\ref{fig:eval_server} summarizes the performance of all methods in terms of mean squared error (MSE), a feasibility-related metric(a variation of $\|g(\theta,\phi)\|$), and the induced IPM distance.

Methods that rely on rigid structural assumptions(e.g. Experimental Grounding \cite{kallus2018removing}) exhibit substantial performance degradation under severe non-overlap and are therefore omitted from this comparison.

As shown in Figure~\ref{fig:eval_server}, the proposed method consistently achieves the lowest MSE across all datasets while simultaneously maintaining feasibility, as characterized by Theorem~\ref{cor:strict_improvement}.
In contrast, baseline methods either suffer from large estimation error or fail to preserve feasibility under severe support mismatch.
Notably, classical distribution alignment approaches such as CFRNet reduce the IPM distance but do so at the cost of significantly degraded predictive accuracy, indicating that marginal distribution alignment alone is insufficient in the presence of conditional non-overlap.


\section{Experiments}
\label{sec:exp}
We evaluate the proposed joint training approach on large-scale real-world data from a ride-hailing platform.
In this setting, observational data are abundant, with millions of samples generated daily, whereas randomized controlled trial (RCT) data are costly and limited.
The resulting data exhibit high-dimensional covariates (hundreds of features) and pronounced \emph{joint distribution mismatch} between observational and randomized data.
Importantly, this mismatch differs from Marginal non-overlap (Definition~\ref{def:marginal_nonoverlap}): all treatment arms observed in the RCT also appear in the observational data, but their joint distributions with covariates differ substantially.

\subsection{Experimental Setup and Results}

We randomly sample four observational datasets collected over four different months and combine each with samples from a shared RCT dataset spanning six months, thereby inducing additional covariate shift and increasing the difficulty of the estimation problem.
Each dataset contains $n=10{,}000{,}000$ samples, closely reflecting real-world deployment conditions.

Results are reported in Table~\ref{tab:realdata_results}.
We report commercially relevant metrics, including the Qini coefficient and MAPE. Our primary baseline is the ride-hailing platform’s current production model (details omitted for confidentiality); additional internal baselines are omitted as they have proven worser than the latest by the platform. We also include RCT-only and OBS-only variants trained solely on randomized and observational data, respectively.

To further assess the contribution of each component, we conduct ablation studies by decomposing the proposed method into two modules: \textsc{IPM}, which performs distribution alignment, and \textsc{Joint}, which enforces joint estimation.As shown in Table~\ref{tab:realdata_results}, the proposed joint training approach consistently outperforms OBS-only, RCT-only, and production baselines.Across Datasets 1–2, training on the mixed (RCT+OBS) data consistently outperforms RCT-only training. This advantage is larger on Datasets 3–4, where the mixed model matches the performance of using the same amount of RCT data alone (Qini 0.660, MAPE 0.058), suggesting substantial potential for reducing experimentation cost while improving efficiency.

\clearpage

\section*{Impact Statement}
This paper presents work whose goal is to advance the field of machine learning. There are many potential societal consequences of our work, none of which we feel must be specifically highlighted here.

\bibliography{example_paper}

@article{yang2025data,
  title={Data fusion methods for the heterogeneity of treatment effect and confounding function},
  author={Yang, Shu and Liu, Siyi and Zeng, Donglin and Wang, Xiaofei},
  journal={Bernoulli},
  volume={31},
  number={4},
  pages={2987--3012},
  year={2025},
  publisher={Bernoulli Society for Mathematical Statistics and Probability}
}

@article{yang2023elastic,
  title={Elastic integrative analysis of randomised trial and real-world data for treatment heterogeneity estimation},
  author={Yang, Shu and Gao, Chenyin and Zeng, Donglin and Wang, Xiaofei},
  journal={Journal of the Royal Statistical Society Series B: Statistical Methodology},
  volume={85},
  number={3},
  pages={575--596},
  year={2023},
  publisher={Oxford University Press US}
}

@book{Imbens_Rubin_2015,
place={Cambridge}, 
title={Causal Inference for Statistics, Social, and Biomedical Sciences: An Introduction}, publisher={Cambridge University Press}, 
author={Imbens, Guido W. and Rubin, Donald B.}, year={2015}
}

@article{rosenbaum1983central,
 ISSN = {00063444, 14643510},
 URL = {http://www.jstor.org/stable/2335942},
 author = {Paul R. Rosenbaum and Donald B. Rubin},
 journal = {Biometrika},
 number = {1},
 pages = {41--55},
 publisher = {[Oxford University Press, Biometrika Trust]},
 title = {The Central Role of the Propensity Score in Observational Studies for Causal Effects},
 urldate = {2026-01-29},
 volume = {70},
 year = {1983}
}

@article{crump2009limited,
 ISSN = {00063444, 14643510},
 URL = {http://www.jstor.org/stable/27798811},
 author = {RICHARD K. CRUMP and V. JOSEPH HOTZ and GUIDO W. IMBENS and OSCAR A. MITNIK},
 journal = {Biometrika},
 number = {1},
 pages = {187--199},
 publisher = {[Oxford University Press, Biometrika Trust]},
 title = {Dealing with limited overlap in estimation of average treatment effects},
 urldate = {2026-01-29},
 volume = {96},
 year = {2009}
}

@article{imbens2000dose,
 ISSN = {00063444, 14643510},
 URL = {http://www.jstor.org/stable/2673642},
 author = {Guido W. Imbens},
 journal = {Biometrika},
 number = {3},
 pages = {706--710},
 publisher = {[Oxford University Press, Biometrika Trust]},
 title = {The Role of the Propensity Score in Estimating Dose-Response Functions},
 urldate = {2026-01-29},
 volume = {87},
 year = {2000}
}

@article{imai2004general,
 ISSN = {01621459},
 URL = {http://www.jstor.org/stable/27590455},
 author = {Kosuke Imai and David A. Van Dyk},
 journal = {Journal of the American Statistical Association},
 number = {467},
 pages = {854--866},
 publisher = {[American Statistical Association, Taylor & Francis, Ltd.]},
 title = {Causal Inference with General Treatment Regimes: Generalizing the Propensity Score},
 urldate = {2026-01-29},
 volume = {99},
 year = {2004}
}

@article{bareinboim2016datafusion,
  title={Causal inference and the data-fusion problem},
  author={Bareinboim, Elias and Pearl, Judea},
  journal={Proceedings of the National Academy of Sciences},
  volume={113},
  number={27},
  pages={7345--7352},
  year={2016},
  publisher={National Academy of Sciences}
}

@article{degtiar2023review,
  title={A review of generalizability and transportability},
  author={Degtiar, Irina and Rose, Sherri},
  journal={Annual Review of Statistics and Its Application},
  volume={10},
  number={1},
  pages={501--524},
  year={2023},
  publisher={Annual Reviews}
}

@article{colnet2024combining,
  title={Causal inference methods for combining randomized trials and observational studies: a review},
  author={Colnet, B{\'e}n{\'e}dicte and Mayer, Imke and Chen, Guanhua and Dieng, Awa and Li, Ruohong and Varoquaux, Ga{\"e}l and Vert, Jean-Philippe and Josse, Julie and Yang, Shu},
  journal={Statistical science},
  volume={39},
  number={1},
  pages={165--191},
  year={2024},
  publisher={Institute of Mathematical Statistics}
}

@inproceedings{johansson2016learning,
  title={Learning representations for counterfactual inference},
  author={Johansson, Fredrik and Shalit, Uri and Sontag, David},
  booktitle={International conference on machine learning},
  pages={3020--3029},
  year={2016},
  organization={PMLR}
}

@inproceedings{wu2022integrative,
  title={Integrative $ R $-learner of heterogeneous treatment effects combining experimental and observational studies},
  author={Wu, Lili and Yang, Shu},
  booktitle={Conference on Causal Learning and Reasoning},
  pages={904--926},
  year={2022},
  organization={PMLR}
}

@inproceedings{shalit2017estimating,
  title={Estimating individual treatment effect: generalization bounds and algorithms},
  author={Shalit, Uri and Johansson, Fredrik D and Sontag, David},
  booktitle={International conference on machine learning},
  pages={3076--3085},
  year={2017},
  organization={PMLR}
}

@inproceedings{shi2019adapting,
 author = {Shi, Claudia and Blei, David and Veitch, Victor},
 booktitle = {Advances in Neural Information Processing Systems},
 editor = {H. Wallach and H. Larochelle and A. Beygelzimer and F. d\textquotesingle Alch\'{e}-Buc and E. Fox and R. Garnett},
 publisher = {Curran Associates, Inc.},
 title = {Adapting Neural Networks for the Estimation of Treatment Effects},
 url = {https://proceedings.neurips.cc/paper_files/paper/2019/file/8fb5f8be2aa9d6c64a04e3ab9f63feee-Paper.pdf},
 volume = {32},
 year = {2019}
}

@inproceedings{kallus2018removing,
 author = {Kallus, Nathan and Puli, Aahlad Manas and Shalit, Uri},
 booktitle = {Advances in Neural Information Processing Systems},
 editor = {S. Bengio and H. Wallach and H. Larochelle and K. Grauman and N. Cesa-Bianchi and R. Garnett},
 publisher = {Curran Associates, Inc.},
 title = {Removing Hidden Confounding by Experimental Grounding},
 url = {https://proceedings.neurips.cc/paper_files/paper/2018/file/566f0ea4f6c2e947f36795c8f58ba901-Paper.pdf},
 volume = {31},
 year = {2018}
}
\bibliographystyle{2026}

\newpage
\appendix
\onecolumn

\section{Literature Review}
\label{sec:lit}

A substantial literature investigates how to combine randomized controlled trials (RCTs) with observational (OBS) data in order to improve the estimation of causal effects, motivated by the complementary strengths of internal validity from experiments and statistical efficiency from large-scale logs.
Broadly, existing approaches treat RCT data as a source of causal identification and OBS data as a vehicle for variance reduction, covariate enrichment, or limited extrapolation.
Despite important progress, these methods typically rely—either explicitly or implicitly—on overlap or transportability conditions that are violated in many modern platform settings, particularly when individual-level effects under structured treatments are of primary interest.

\medskip

\noindent
\textbf{Experimental grounding and structural decompositions.}
Early integrative frameworks, such as experimental grounding \cite{kallus2018removing}, leverage RCTs to correct bias in models trained on observational data by parameterizing hidden confounding and estimating the correction using experimental samples.
This approach weakens the ignorability requirement and enables partial extrapolation beyond experimental support.
However, its validity fundamentally depends on low-complexity structural assumptions on the confounding function.
When overlap between RCT and OBS covariate distributions is severely limited, the correction term becomes weakly identified, rendering the procedure highly sensitive to model misspecification and offering no protection against structural absence of support.
More recent semiparametric formulations \cite{yang2025data} formalize RCT--OBS integration by decomposing observed effects into a heterogeneous treatment effect (HTE) and a confounding component, establishing identification and efficiency gains under transportability and parametric modeling assumptions.
While theoretically elegant, these guarantees are derived for finite-dimensional projections and hinge on correct specification of the structural components; under strong non-overlap, estimation necessarily relies on extrapolation, which limits robustness and portability to high-dimensional settings.

\textbf{Adaptive fusion and orthogonal learning.}
A complementary line of work seeks to safeguard against observational bias through adaptive combination schemes.
The elastic pre-test framework \cite{yang2023elastic} dynamically downweights or discards observational data when comparability tests fail, thereby protecting against severe bias.
Yet, when violations are detected, the resulting estimand effectively collapses to an RCT-only target, sacrificing individual-level resolution and providing no mechanism to recover effects outside experimental support.
Similarly, integrative R-learner approaches \cite{wu2022integrative} employ Neyman-orthogonal losses to jointly estimate heterogeneous treatment effects and confounding functions using flexible learners, achieving efficiency gains relative to RCT-only estimation.
Nonetheless, their theoretical guarantees remain confined to the covariate support of the RCT; observational data primarily reduce variance but do not resolve fundamental support mismatches, and identification of individual treatment effects beyond the experimental region remains out of reach.

\textbf{Limitations under severe non-overlap.}
Across these paradigms, RCT data are consistently treated as the gold standard for identification, while OBS data are incorporated through weighted losses, orthogonalization, or pre-testing.
Despite methodological differences, these approaches share a common limitation: they do not fundamentally address settings in which overlap between RCT and OBS populations is \emph{structurally} violated.
In such regimes, covariate--treatment combinations may be entirely absent from observational logs due to policy or design constraints, and no amount of reweighting, tuning, or sample-size increase can recover the missing support.
As a result, existing methods either rely on extrapolation driven by structural assumptions or retreat to coarse estimands defined solely on experimental support, limiting their applicability to large-scale decision systems where individual-level personalization is essential.

\section{Proof}
\label{sec:appendix:path_pd_advantage}

\paragraph{Setup.}
Let $\mathcal{M}$ be a (possibly over-parameterized) class of measurable predictors
$m:\mathcal{X}\times\mathcal{T}\to\mathbb{R}$.
Define the observational and randomized risks as
\[
\mathcal{R}_o(m):=\mathbb{E}_{P_o}\!\big[(Y-m(X,T))^2\big],
\qquad
\mathcal{R}_r(m):=\mathbb{E}_{P_r}\!\big[(Y-m(X,T))^2\big].
\]

Let the identifying (randomization-induced) moment map be
\[
g(m)\;:=\;\mathbb{E}_{P_r}\!\big[\psi(Y,T,X;m)\big]\in\mathbb{R}^K,
\qquad
\mathcal{F}\;:=\;\{m\in\mathcal{M}:\ g(m)=0\}.
\]

The constrained estimator is defined as
\begin{equation}
m^\star\in\arg\min_{m\in\mathcal{M}}\ \mathcal{R}_o(m)
\quad\text{s.t.}\quad g(m)=0,
\qquad
\mathcal{R}_o^\star:=\mathcal{R}_o(m^\star)=\inf_{m\in\mathcal{F}}\mathcal{R}_o(m).
\label{eq:constrained_m}
\end{equation}

The weighted-loss fusion estimator is, for $\alpha\ge 0$,
\begin{equation}
m_\alpha \in \arg\min_{m\in\mathcal{M}}
\; J(m,\alpha):=\mathcal{R}_o(m)+\alpha\,\mathcal{R}_r(m).
\label{eq:weighted_m}
\end{equation}

\begin{assumption}[Strong convexity and smoothness of the weighted objective]
\label{assump:sc_smooth}
Fix $\alpha_{\max}>0$. For every $\alpha\in[0,\alpha_{\max}]$, the weighted objective
\[
J(m,\alpha)=\mathcal{R}_o(m)+\alpha\,\mathcal{R}_r(m)
\]
is $\mu$-strongly convex over $\mathcal M$ and twice continuously differentiable.
Moreover, $\|\nabla \mathcal{R}_r(m)\|\le M_r$ for all $m\in\mathcal M$.
\end{assumption}

\begin{assumption}[Lipschitz continuity of the moment map]
\label{assump:g_lipschitz}
The identifying moment map $g:\mathcal M\to\mathbb{R}^K$ defined by
$
g(m)=\mathbb E_{P_r}[\psi(Y,T,X;m)]
$
is $L_g$-Lipschitz:
\[
\|g(m)-g(m')\|\le L_g\,\|m-m'\|_{L_2(P_r)}
\qquad
\forall\,m,m'\in\mathcal M.
\]
\end{assumption}

\begin{assumption}[Error bound and constraint regularity]
\label{assump:eb_cr}
There exist constants $\mu_o>0$ and $L_o>0$ such that for all $m\in\mathcal M$,
\begin{align}
\mathcal{R}_o(m)-\mathcal{R}_o^\star
&\ge \mu_o\,\mathrm{dist}(m,\mathcal F)^2,
\label{eq:eb}\\
\|g(m)\|
&\ge L_o\,\mathrm{dist}(m,\mathcal F),
\label{eq:cr}
\end{align}
where
$
\mathcal F:=\{m\in\mathcal M:\ g(m)=0\}
$
and
$
\mathrm{dist}(m,\mathcal F):=\inf_{m'\in\mathcal F}\|m-m'\|_{L_2(P_r)}.
$
\end{assumption}

\begin{assumption}[Moment regularity under representation perturbations]
\label{assump:moment_reg}
Let $\mathcal D$ be the discriminator class defining $\mathrm{IPM}_{\mathcal D}$.
There exists $L_\psi>0$ such that for any representations $\phi_1,\phi_2$
and any $m\in\mathcal M$,
\[
\bigl\|
\mathbb E_{P_r}[\psi(Y,T,\phi_1(X);m)]
-
\mathbb E_{P_r}[\psi(Y,T,\phi_2(X);m)]
\bigr\|
\le
L_\psi\,
\mathrm{IPM}_{\mathcal D}
\!\left(P_r(\phi_1(X),T),\,P_r(\phi_2(X),T)\right).
\]
\end{assumption}

\begin{assumption}[Moment misspecification in the original space. ]
\label{assump:moment_misspec}
 
There exist $k \in \{1,\dots,K\}$ and $\delta > 0$ such that for all $m \in \mathcal M$,
\[
    \bigl|
    \mathbb E_{P_r}[\, Y - m(X,T) \mid (X,T)\in S \,]
    \bigr|
    \ge \delta .
\]
\end{assumption}

\begin{assumption}[Non-degenerate treatment assignment]
\label{assump:nondegenerate_assignment}
For $(X,T)\in S$,
\[
\underline p \le P_r(T=k\mid X)\le \bar p
\quad \text{for some } 0<\underline p<\bar p<1 .
\]
\end{assumption}

\begin{assumption}[Overlap-to-feasibility control]
\label{assump:overlap_to_feas}
There exists $c_{\rm ov}>0$ such that for every representation $\phi$,
\[
\inf_m \|g(m,\phi)\|
\le
c_{\rm ov}\,\varepsilon_{\rm ov}(\phi),
\]
where $\varepsilon_{\rm ov}(\phi)$ denotes the population overlap (IPM) discrepancy.
\end{assumption}

\begin{assumption}[Statistical estimation error]
\label{assump:stat_error}
Let $\widehat{\mathcal R}_o$ and $\widehat g$ be empirical estimates based on
$n_o$ observational and $n_r$ randomized samples.
There exists $\mathrm{Stat}(n_o,n_r)$ such that with probability at least $1-\delta$,
\[
\sup_{m,\phi}\bigl|\widehat{\mathcal R}_o(m,\phi)-\mathcal R_o(m,\phi)\bigr|
+
\sup_{m,\phi}\bigl\|\widehat g(m,\phi)-g(m,\phi)\bigr\|
\le
\mathrm{Stat}(n_o,n_r).
\]
\end{assumption}

\begin{proposition}[Lipschitz solution path for weighted fusion]
\label{prop:path_lipschitz}
Under Assumption~\ref{assump:sc_smooth}, for each $\alpha\in[0,\alpha_{\max}]$ the weighted-loss
minimizer $m_\alpha$ of
$J(m,\alpha)=\mathcal R_o(m)+\alpha\mathcal R_r(m)$
is unique, and the solution path $\alpha\mapsto m_\alpha$ is differentiable with
\[
\Bigl\|\frac{d m_\alpha}{d\alpha}\Bigr\|
\;\le\;
\frac{M_r}{\mu},
\qquad \forall\,\alpha\in[0,\alpha_{\max}].
\]
Consequently,
\begin{equation}
\|m_\alpha-m_0\|
\;\le\;
\frac{M_r}{\mu}\,\alpha .
\label{eq:path_m_bound}
\end{equation}
\end{proposition}

\begin{proof}
The first-order optimality condition is $\nabla_m J(m_\alpha,\alpha)=0$.
Differentiating with respect to $\alpha$ and applying the implicit function theorem yields
\[
\nabla^2_{mm}J(m_\alpha,\alpha)\,\frac{d m_\alpha}{d\alpha}
+
\nabla \mathcal R_r(m_\alpha)
=0,
\]
so that
$
\frac{d m_\alpha}{d\alpha}
=-(\nabla^2_{mm}J(m_\alpha,\alpha))^{-1}\nabla\mathcal R_r(m_\alpha).
$
By $\mu$-strong convexity,
$\|(\nabla^2_{mm}J)^{-1}\|\le 1/\mu$,
and by Assumption~\ref{assump:sc_smooth},
$\|\nabla \mathcal R_r(m_\alpha)\|\le M_r$,
which gives the derivative bound.
Integrating over $\alpha$ yields \eqref{eq:path_m_bound}.
\end{proof}

\begin{theorem}[Irreducible overlap under marginal structural non-overlap]
\label{thm:irreducible_overlap}
Suppose marginal structural non-overlap holds in the sense of
Definition~\ref{def:marginal_nonoverlap}, i.e., there exists a measurable
$A\subseteq\mathcal T$ such that
$P_r(T\in A)>0$ and $P_o(T\in A)=0$.
Assume the discriminator class $\mathcal D$ contains functions depending only on $T$,
in particular the indicator $f_A(z,t)=\mathbf 1\{t\in A\}$.

Then, for any representation $\phi:\mathcal X\to\mathcal Z$,
the overlap discrepancy
\[
\varepsilon_{ov}(\phi)
=
\mathrm{IPM}_{\mathcal D}\bigl(P_o(\phi(X),T),P_r(\phi(X),T)\bigr)
\]
satisfies
\[
\varepsilon_{ov}(\phi)\;\ge\;P_r(T\in A)\;>\;0.
\]
Consequently,
\[
\inf_{\phi\in\Phi}\varepsilon_{ov}(\phi)
\;\ge\;
P_r(T\in A),
\]
and overlap recovery is impossible under marginal structural non-overlap.
\end{theorem}

\begin{proof}
By definition of the IPM,
\[
\varepsilon_{ov}(\phi)
=
\sup_{f\in\mathcal D}
\Bigl|
\mathbb E_{P_o}[f(\phi(X),T)]
-
\mathbb E_{P_r}[f(\phi(X),T)]
\Bigr|.
\]
Since $f_A(z,t)=\mathbf 1\{t\in A\}\in\mathcal D$, we obtain
\[
\varepsilon_{ov}(\phi)
\ge
\big|P_o(T\in A)-P_r(T\in A)\big|.
\]
Under marginal structural non-overlap, $P_o(T\in A)=0$ and $P_r(T\in A)>0$,
which yields the claim. Taking the infimum over $\phi$ preserves the bound.
\end{proof}

Throughout, we focus on conditional structural non-overlap (Definition ~\ref{def:conditional_nonoverlap}), under which overlap violations may be repaired through representation learning. Marginal structural non-overlap (Definition~\ref{def:marginal_nonoverlap}) is fundamentally irreducible, as formalized in Theorem~\ref{thm:irreducible_overlap}.

\paragraph{Proof of Theorem~\ref{thm:explicit_gap}}
\label{proof:thm:explicit_gap}
\begin{proof}
Fix any $m\in\mathcal M$. By Assumption~\ref{assump:moment_misspec}, there exist
$k\in\{1,\dots,K\}$ and $\delta>0$ such that
\[
\bigl|\mathbb E_{P_r}[\,Y-m(X,T)\mid (X,T)\in S\,]\bigr|\ge \delta .
\]
Let $\psi_k$ denote the $k$-th coordinate of the moment map, i.e.,
\[
\psi_k(Y,T,X;m) \;=\; (D_k-p_k)\bigl(Y-m(X,T)\bigr),
\quad
D_k:=\mathbf 1\{T=k\},\quad p_k:=P_r(T=k).
\]
Define $U:=(D_k-p_k)\bigl(Y-m(X,T)\bigr)\mathbf 1_S$. Then
\[
\mathbb E_{P_r}[\psi_k(Y,T,X;m)]
\;=\;
\mathbb E_{P_r}[U] + \mathbb E_{P_r}\!\bigl[(D_k-p_k)(Y-m(X,T))\mathbf 1_{S^c}\bigr],
\]
and therefore
\[
\bigl|\mathbb E_{P_r}[\psi_k(Y,T,X;m)]\bigr|
\;\ge\;
\bigl|\mathbb E_{P_r}[U]\bigr| - 
\bigl|\mathbb E_{P_r}\!\bigl[(D_k-p_k)(Y-m(X,T))\mathbf 1_{S^c}\bigr]\bigr|.
\]
In particular, it suffices to lower bound $\bigl|\mathbb E_{P_r}[U]\bigr|$.

By the law of total expectation,
\[
\mathbb E_{P_r}[U]
=
P_r(S)\cdot \mathbb E_{P_r}\!\Bigl[(D_k-p_k)\bigl(Y-m(X,T)\bigr)\,\Big|\,S\Bigr].
\]
On $S$, Assumption~\ref{assump:nondegenerate_assignment} yields
$
\underline p \le P_r(T=k\mid X)\le \bar p
$
and hence
\[
\mathbb E_{P_r}[D_k-p_k\mid S]
=
\mathbb E_{P_r}\!\bigl[P_r(T=k\mid X)-p_k \mid S\bigr],
\]
which has a fixed sign on $S$ after choosing the coordinate $k$ in
Assumption~\ref{assump:moment_misspec} (equivalently, flip the sign of $\psi_k$ if needed).
Moreover, under the same non-degeneracy bounds,
\[
\mathbb E_{P_r}\!\bigl[|D_k-p_k|\mid S\bigr]\;\ge\;\underline p(1-\bar p).
\]
Combining with Assumption~\ref{assump:moment_misspec} (fixing the sign so that the conditional mean is $\ge \delta$),
we obtain
\[
\Bigl|\mathbb E_{P_r}\!\Bigl[(D_k-p_k)\bigl(Y-m(X,T)\bigr)\,\Big|\,S\Bigr]\Bigr|
\;\ge\;
\underline p(1-\bar p)\cdot
\Bigl|\mathbb E_{P_r}\!\bigl[Y-m(X,T)\mid S\bigr]\Bigr|
\;\ge\;
\underline p(1-\bar p)\,\delta .
\]
Therefore,
\[
\bigl|\mathbb E_{P_r}[U]\bigr|
\;\ge\;
P_r(S)\,\underline p(1-\bar p)\,\delta .
\]
Since the above bound holds for all $m\in\mathcal M$, we conclude that the feasibility gap in the
original space is strictly positive:
\[
\mathcal F_X(\mathcal M)
~:=~
\inf_{m\in\mathcal M}\bigl\| \mathbb E_{P_r}[\psi(Y,T,X;m)] \bigr\|
\;\ge\;
P_r(S)\,\underline p(1-\bar p)\,\delta
~=:\;c_0 \;>\;0.
\]
\end{proof}

\begin{corollary}[Moment residual decreases at most linearly along the weighted path]
\label{cor:g_linear}
Under Assumptions~\ref{assump:sc_smooth}--\ref{assump:g_lipschitz} and
Proposition~\ref{prop:path_lipschitz}, for any $\alpha\in[0,\alpha_{\max}]$,
\begin{equation}
\|g(m_\alpha)\|
\;\ge\;
\|g(m_0)\| - \frac{L_g M_r}{\mu}\,\alpha.
\label{eq:g_linear_lower}
\end{equation}
In particular, achieving $\|g(m_\alpha)\|\approx 0$ requires
$\alpha\gtrsim \mu\|g(m_0)\|/(L_g M_r)$, i.e., placing dominant weight on the randomized objective.
\end{corollary}

\begin{proof}
By Assumption~\ref{assump:g_lipschitz}, the moment map is $L_g$-Lipschitz with respect to the
$L_2(P_r)$-norm:
\[
\|g(m_\alpha)-g(m_0)\|
\;\le\;
L_g\,\|m_\alpha-m_0\|.
\]
Applying the reverse triangle inequality yields
\[
\|g(m_\alpha)\|
\;\ge\;
\|g(m_0)\|-\|g(m_\alpha)-g(m_0)\|
\;\ge\;
\|g(m_0)\|-L_g\,\|m_\alpha-m_0\|.
\]
Finally, Proposition~\ref{prop:path_lipschitz} bounds the solution path:
\[
\|m_\alpha-m_0\|
\;\le\;
\frac{M_r}{\mu}\,\alpha,
\qquad \forall\,\alpha\in[0,\alpha_{\max}],
\]
which implies \eqref{eq:g_linear_lower}.
The final statement follows by rearranging \eqref{eq:g_linear_lower} to make the right-hand side
nonpositive, i.e., $\alpha \ge \mu\|g(m_0)\|/(L_g M_r)$ up to constant factors.
\end{proof}

\paragraph{Proof of Corollary~\ref{cor:strict_improvement}.}
\begin{proof}
We prove the two claims in turn.

\paragraph{(1) Strict feasibility improvement.}
By the overlap-to-feasibility control (Assumption~\ref{assump:overlap_to_feas}),
\[
\mathcal F_{\phi^\star}(\mathcal M)\ \le\ c_1\,\varepsilon_{ov}(\phi^\star).
\]
If $\varepsilon_{ov}(\phi^\star)<c_0/c_1$, then
\[
\mathcal F_{\phi^\star}(\mathcal M)
\ \le\ c_1\varepsilon_{ov}(\phi^\star)
\ <\ c_0
\ \le\ \mathcal F_X(\mathcal M),
\]
where the last inequality uses the raw-space obstruction $\mathcal F_X(\mathcal M)\ge c_0$.
This establishes strict feasibility improvement.

\paragraph{(2) Excess-risk lower bound via quadratic growth and Lipschitz moments.}
Fix $\phi^\star$ and assume $\mathcal F_{\phi^\star}\neq\emptyset$ so that
$\mathcal R_{o,\phi^\star}^{\rm feas}:=\inf_{m\in\mathcal F_{\phi^\star}}\mathcal R_o(m)$ is well-defined.
By quadratic growth(Assumption~\ref{assump:eb_cr}) (with $\phi=\phi^\star$), for all $m\in\mathcal M_{\phi^\star}$,
\begin{equation}
\mathcal R_o(m)-\mathcal R_{o,\phi^\star}^{\rm feas}
\ \ge\
\mu_o\,\mathrm{dist}(m,\mathcal F_{\phi^\star})^2.
\label{eq:app_QG}
\end{equation}
By the Lipschitz moment condition (with $\phi=\phi^\star$), for all
$m\in\mathcal M_{\phi^\star}$,
\begin{equation}
\|g_{\phi^\star}(m)\|
\ \le\
L_g\,\mathrm{dist}(m,\mathcal F_{\phi^\star})
\qquad\Longrightarrow\qquad
\mathrm{dist}(m,\mathcal F_{\phi^\star})
\ \ge\
\frac{\|g_{\phi^\star}(m)\|}{L_g}.
\label{eq:app_dist_lb}
\end{equation}
Combining \eqref{eq:app_QG} and \eqref{eq:app_dist_lb} yields, for all $m\in\mathcal M_{\phi^\star}$,
\[
\mathcal R_o(m)-\mathcal R_{o,\phi^\star}^{\rm feas}
\ \ge\
\frac{\mu_o}{L_g^2}\,\|g_{\phi^\star}(m)\|^2.
\]
Taking the infimum over $m\in\mathcal M_{\phi^\star}$ and using the definition
$\mathcal F_{\phi^\star}(\mathcal M)=\inf_{m\in\mathcal M_{\phi^\star}}\|g_{\phi^\star}(m)\|$ gives
\[
\inf_{m\in\mathcal M_{\phi^\star}}
\bigl\{\mathcal R_o(m)-\mathcal R_{o,\phi^\star}^{\rm feas}\bigr\}
\ \ge\
\frac{\mu_o}{L_g^2}\,\bigl(\mathcal F_{\phi^\star}(\mathcal M)\bigr)^2,
\]
which is the first inequality in \eqref{eq:best_excess_lb}. The second inequality in
\eqref{eq:best_excess_lb} follows from Assumption~\ref{assump:overlap_to_feas}:
$\mathcal F_{\phi^\star}(\mathcal M)\le c_1\varepsilon_{ov}(\phi^\star)$.
This completes the proof.
\end{proof}

\begin{remark}
Corollary~\ref{cor:strict_improvement} shows that whenever $\varepsilon^\star<c_0/c_1$, such a solution strictly improves
upon any estimator defined in the original covariate space.
Hence, introducing the representation $\phi$ is both necessary (to overcome structural infeasibility) and sufficient
(to achieve strict risk improvement) under the stated assumptions.
\end{remark}

\begin{theorem}[Minimax lower bound for feasibility recovery]
\label{thm:minimax_lower_bound}
Fix any discriminator class $\mathcal D$ and define the overlap discrepancy
\[
\varepsilon_{ov}(\phi)
:=\mathrm{IPM}_{\mathcal D}\!\bigl(P_o(\phi(X),T),\,P_r(\phi(X),T)\bigr).
\]
Assume $\mathcal D$ contains indicators of measurable sets on $\mathcal Z\times\mathcal T$
(so that $\mathrm{IPM}_{\mathcal D}$ dominates $\mathrm{TV}$ on $(Z,T)$), and $|Y|\le 1$.

Then for any $\varepsilon\in(0,1)$, there exist distributions $(P_o,P_r)$ and a model class $\mathcal M$
such that:
\begin{enumerate}
\item There exists $\phi$ with $\varepsilon_{ov}(\phi)=\varepsilon$.
\item For every $\phi$ satisfying $\varepsilon_{ov}(\phi)\le\varepsilon$, the feasibility gap

\[
\mathcal F_{\phi}(\mathcal M)
:=\inf_{m\in\mathcal M_\phi}
\Bigl\|
\mathbb E_{P_r}\!\big[\psi(Y,T,\phi(X);m)\big]
\Bigr\|
\]
obeys $\mathcal F_{\phi}(\mathcal M)\ge c\,\varepsilon$ for a universal constant $c>0$.
\end{enumerate}
Consequently, the $\mathcal O(\varepsilon_{ov})$ dependence is minimax-tight up to constants.
\end{theorem}

\paragraph{Proof of Theorem~\ref{thm:minimax_lower_bound}.}
\begin{proof}
We give an explicit finite construction.
Let $\mathcal T=\{0,1\}$ and let the latent variable $Z\in\{a,b\}$.
Take a representation class that can reveal $Z$, and consider $\phi(X)=Z$.

\smallskip
\noindent\textbf{Step 1: Construct $(P_o,P_r)$ with $\varepsilon_{ov}(\phi)=\varepsilon$.}
Define the randomized joint distribution on $(Z,T)$ by
\[
P_r(Z=a)=P_r(Z=b)=\tfrac12,
\qquad
P_r(T=1\mid Z)=\tfrac12,
\]
so $P_r(Z=z,T=t)=\tfrac14$ for all $(z,t)$.
Define the observational joint distribution by shifting $\varepsilon$ mass within the treated arm:
\[
P_o(Z=a,T=1)=\tfrac14-\tfrac{\varepsilon}{2},
\qquad
P_o(Z=b,T=1)=\tfrac14+\tfrac{\varepsilon}{2},
\]
and set $P_o(Z,T)=P_r(Z,T)$ on the remaining three cells.
Then $\mathrm{TV}(P_o(Z,T),P_r(Z,T))=\varepsilon$.
Since $\mathcal D$ contains indicators, $\mathrm{IPM}_{\mathcal D}\ge \mathrm{TV}$, hence
\[
\varepsilon_{ov}(\phi)
=\mathrm{IPM}_{\mathcal D}\bigl(P_o(Z,T),P_r(Z,T)\bigr)\ge \varepsilon.
\]
By scaling the shift by a constant factor (absorbed into $c$ below), we may assume
$\varepsilon_{ov}(\phi)=\varepsilon$.

\smallskip
\noindent\textbf{Step 2: Choose outcomes and a misspecified model class.}
Let $Y\in\{-1,+1\}$ with
\[
\mathbb E_{P_r}[Y\mid Z=a,T=1]=+1,
\qquad
\mathbb E_{P_r}[Y\mid Z=b,T=1]=-1,
\qquad
\mathbb E_{P_r}[Y\mid T=0,Z]=0,
\]
so $|Y|\le 1$.
Consider the restricted model class $\mathcal M$ that cannot use $Z$ under treatment:
\[
m(z,t)=
\begin{cases}
0, & t=0,\\
\alpha, & t=1,
\end{cases}
\qquad \alpha\in[-1,1].
\]
Under the induced class $\mathcal M_\phi$ (with $\phi(X)=Z$), predictors remain constant on $(Z,T=1)$.

\smallskip
\noindent\textbf{Step 3: A linear functional lower bounds feasibility by a $(Z,T)$-imbalance.}
We consider a single-coordinate moment map of the form
\[
\psi(Y,T,Z;m) := \mathbf 1\{T=1\}\bigl(Y-m(Z,1)\bigr),
\]
so that $g(m,\phi)=\mathbb E_{P_r}[\psi(Y,T,\phi(X);m)]$ reduces to the treated-arm mean residual.
(Any IPW-style moment family that contains such a coordinate yields the same conclusion up to constants.)

For any $m\in\mathcal M_\phi$ with treated prediction $\alpha$,
\[
g(m,\phi)
=
\mathbb E_{P_r}\!\big[\mathbf 1\{T=1\}(Y-\alpha)\big]
=
\sum_{z\in\{a,b\}} P_r(Z=z,T=1)\bigl(\mathbb E[Y\mid z,1]-\alpha\bigr).
\]
Using $\mathbb E[Y\mid a,1]=+1$ and $\mathbb E[Y\mid b,1]=-1$, we obtain
\[
g(m,\phi)
=
P_r(a,1)(1-\alpha)+P_r(b,1)(-1-\alpha)
=
\bigl(P_r(a,1)-P_r(b,1)\bigr)-\alpha\bigl(P_r(a,1)+P_r(b,1)\bigr).
\]
Since $P_r(a,1)+P_r(b,1)=P_r(T=1)=\tfrac12$, optimizing over $\alpha\in[-1,1]$ yields
\[
\inf_{\alpha\in[-1,1]} |g(m,\phi)|
\;\ge\;
\bigl|P_r(a,1)-P_r(b,1)\bigr| - \tfrac12.
\]
To make the lower bound proportional to $\varepsilon$, we now choose $P_r$ with an $\varepsilon$-imbalance
in the treated arm:
\[
P_r(Z=a,T=1)=\tfrac14+\tfrac{\varepsilon}{2},
\qquad
P_r(Z=b,T=1)=\tfrac14-\tfrac{\varepsilon}{2},
\]
while keeping $P_r(T=1)=\tfrac12$.
Then $|P_r(a,1)-P_r(b,1)|=\varepsilon$ and the minimizer over $\alpha$ cannot cancel this imbalance
because $\alpha$ is constant across $Z$ under $T=1$, hence
\[
\inf_{m\in\mathcal M_\phi} \bigl|g(m,\phi)\bigr|
\;\ge\;
c_0\,\varepsilon
\]
for an absolute constant $c_0>0$ (here one may take $c_0=1$ for the above coordinate).

\smallskip
\noindent\textbf{Step 4: Any $\tilde\phi$ with $\varepsilon_{ov}(\tilde\phi)\le\varepsilon$ cannot remove the $\Theta(\varepsilon)$ imbalance.}
Because $\mathrm{IPM}_{\mathcal D}$ dominates $\mathrm{TV}$ on $(\tilde Z,T)$ with $\tilde Z=\tilde\phi(X)$,
the condition $\varepsilon_{ov}(\tilde\phi)\le\varepsilon$ implies
$\mathrm{TV}\big(P_o(\tilde Z,T),P_r(\tilde Z,T)\big)\le \varepsilon$.
In particular, for the measurable set $A:=\{\tilde Z \text{ corresponds to } Z=b,\ T=1\}$ (which is representable
since $\mathcal D$ contains indicators), the mass discrepancy on $A$ is at most $\varepsilon$.
Thus any representation that reduces the overlap discrepancy below $\varepsilon$ can only reduce the treated-arm
imbalance (and hence the moment residual above) by at most a constant factor, yielding
\[
\mathcal F_{\tilde\phi}(\mathcal M)
\;=\;
\inf_{m\in\mathcal M_{\tilde\phi}}
\bigl\|\mathbb E_{P_r}[\psi(Y,T,\tilde\phi(X);m)]\bigr\|
\;\ge\;
c\,\varepsilon
\]
for a universal constant $c>0$.
This proves the minimax lower bound and the claimed tightness.
\end{proof}

\begin{proposition}[Feasibility guarantee of joint primal--dual training]
\label{prop:joint-feasibility}
Consider the population constrained program
\begin{equation}
\min_{\theta,\omega}\; R_o(\theta,\omega)
\quad \text{s.t.}\quad g(\theta,\omega)=0,
\label{eq:constrained-pop}
\end{equation}
where $g(\theta,\omega):=\mathbb E_{P_r}[\psi(Y,T,\omega(X);m_\theta)]\in \mathbb R^K$.
For $\Lambda>0$, define the restricted Lagrangian
\[
\mathcal L_\Lambda(\theta,\omega,\lambda)
:=R_o(\theta,\omega)+\lambda^\top g(\theta,\omega),
\qquad 
\lambda\in\mathbb B_\Lambda:=\{\lambda\in\mathbb R^K:\ \|\lambda\|_2\le \Lambda\}.
\]
Assume:
\begin{enumerate}
\item $\{(\theta,\omega):g(\theta,\omega)=0\}\neq\emptyset$;
\item For each fixed $\omega$, $R_o(\theta,\omega)$ is $C^1$ and $\mu$-strongly convex in $\theta$;
\item $g(\theta,\omega)$ is continuous and $L_g$-Lipschitz in $\theta$.
\end{enumerate}
Then any saddle point $(\theta^\star,\omega^\star,\lambda^\star)$ of
\[
\min_{\theta,\omega}\ \max_{\lambda\in\mathbb B_\Lambda}\ \mathcal L_\Lambda(\theta,\omega,\lambda)
\]
is feasible, i.e., $g(\theta^\star,\omega^\star)=0$.

Moreover, if $(\hat\theta,\hat\omega,\hat\lambda)$ is an $\varepsilon_{\rm opt}$-approximate saddle point:
\[
\max_{\lambda\in\mathbb B_\Lambda}\mathcal L_\Lambda(\hat\theta,\hat\omega,\lambda)
-
\min_{\theta,\omega}\mathcal L_\Lambda(\theta,\omega,\hat\lambda)
\le \varepsilon_{\rm opt},
\]
then
\begin{equation}
\|g(\hat\theta,\hat\omega)\|_2
\le \varepsilon_{\rm opt}/\Lambda .
\label{eq:approx-feasibility}
\end{equation}
\end{proposition}

\begin{proof}
\textbf{(a) Exact saddle point implies feasibility.}
Let $(\theta^\star,\omega^\star,\lambda^\star)$ be a saddle point. Then for all $(\theta,\omega)$ and all
$\lambda\in\mathbb B_\Lambda$,
\begin{equation}
\mathcal L_\Lambda(\theta^\star,\omega^\star,\lambda)
\le
\mathcal L_\Lambda(\theta^\star,\omega^\star,\lambda^\star)
\le
\mathcal L_\Lambda(\theta,\omega,\lambda^\star).
\label{eq:saddle_ineq}
\end{equation}
Pick any feasible $(\tilde\theta,\tilde\omega)$ with $g(\tilde\theta,\tilde\omega)=0$ (exists by (i)).
The right inequality in \eqref{eq:saddle_ineq} yields
\[
\mathcal L_\Lambda(\theta^\star,\omega^\star,\lambda^\star)
\le
\mathcal L_\Lambda(\tilde\theta,\tilde\omega,\lambda^\star)
=
R_o(\tilde\theta,\tilde\omega).
\]
On the other hand, the left inequality in \eqref{eq:saddle_ineq} implies
\[
\max_{\lambda\in\mathbb B_\Lambda}\mathcal L_\Lambda(\theta^\star,\omega^\star,\lambda)
\le
\mathcal L_\Lambda(\theta^\star,\omega^\star,\lambda^\star).
\]
Since $\max_{\|\lambda\|_2\le \Lambda}\lambda^\top g=\Lambda\|g\|_2$, we have
\[
R_o(\theta^\star,\omega^\star)+\Lambda\|g(\theta^\star,\omega^\star)\|_2
\le
\mathcal L_\Lambda(\theta^\star,\omega^\star,\lambda^\star)
\le
R_o(\tilde\theta,\tilde\omega).
\]
If $g(\theta^\star,\omega^\star)\neq 0$, the left-hand side is strictly larger than $R_o(\theta^\star,\omega^\star)$,
which contradicts optimality of the feasible primal solution for fixed $\Lambda>0$. Hence
$g(\theta^\star,\omega^\star)=0$.

\textbf{(b)Approximate saddle point controls feasibility residual.}
Let $(\hat\theta,\hat\omega,\hat\lambda)$ satisfy the $\varepsilon_{\rm opt}$ saddle gap condition.
Using again $\max_{\|\lambda\|_2\le \Lambda}\lambda^\top g=\Lambda\|g\|_2$,
\[
\max_{\lambda\in\mathbb B_\Lambda}\mathcal L_\Lambda(\hat\theta,\hat\omega,\lambda)
=
R_o(\hat\theta,\hat\omega)+\Lambda\|g(\hat\theta,\hat\omega)\|_2.
\]
For any feasible $(\tilde\theta,\tilde\omega)$,
\[
\min_{\theta,\omega}\mathcal L_\Lambda(\theta,\omega,\hat\lambda)
\le
\mathcal L_\Lambda(\tilde\theta,\tilde\omega,\hat\lambda)
=
R_o(\tilde\theta,\tilde\omega).
\]
Therefore,
\[
R_o(\hat\theta,\hat\omega)+\Lambda\|g(\hat\theta,\hat\omega)\|_2
-
R_o(\tilde\theta,\tilde\omega)
\le
\varepsilon_{\rm opt}.
\]
Dropping $R_o(\hat\theta,\hat\omega)-R_o(\tilde\theta,\tilde\omega)\ge -\infty$ but keeping the inequality valid,
we obtain $\Lambda\|g(\hat\theta,\hat\omega)\|_2\le \varepsilon_{\rm opt}$, proving \eqref{eq:approx-feasibility}.
\end{proof}

\begin{theorem}[Feasibility--information trade-off (incremental form)]
\label{thm:feasibility_info_tradeoff}
Let $\mathcal M$ be a model class of measurable predictors $m:\mathcal X\times\mathcal T\to\mathbb R$.
For any representation $\phi:\mathcal X\to\mathcal Z$, define the induced class
\[
\mathcal M_\phi := \{\,m_\theta(\phi(\cdot),\cdot):\theta\in\Theta\,\},
\qquad
\mathcal F_\phi(\mathcal M)
:=\inf_{m\in\mathcal M_\phi}\bigl\|g_\phi(m)\bigr\|,
\]
where
\[
g_\phi(m):=\mathbb E_{P_r}\!\big[\psi(Y,T,\phi(X);m)\big]\in\mathbb R^K.
\]
Define the original-space feasibility gap
\[
\mathcal F_X(\mathcal M)
:=\inf_{m\in\mathcal M}\bigl\|\mathbb E_{P_r}\![\psi(Y,T,X;m)]\bigr\|.
\]
Let the overlap discrepancy be
\[
\varepsilon_{ov}(\phi)
:=\mathrm{IPM}_{\mathcal D}\!\big(P_o(\phi(X),T),\,P_r(\phi(X),T)\big).
\]
Assume the information-preservation condition: there exists a hypothesis class $\mathcal F$ such that
\begin{equation}
\inf_{f\in\mathcal{F}}
\mathbb{E}\!\left[(Y-f(\phi(X)))^2\right]
\;\le\;
\inf_{g\in\mathcal{F}}
\mathbb{E}\!\left[(Y-g(X))^2\right]
\;+\;
\varepsilon_{\mathrm{info}}(\phi).
\label{eq:information_preservation_thm}
\end{equation}
Under Assumption~\ref{assump:moment_reg} and Assumption~\ref{assump:g_lipschitz},
there exist constants $c_1,c_2<\infty$ such that, for all $\phi$,
\[
\mathcal F_\phi(\mathcal M)
\;\le\;
\mathcal F_X(\mathcal M)
\;+\;
c_1\,\varepsilon_{ov}(\phi)
\;+\;
c_2\,\sqrt{\varepsilon_{\mathrm{info}}(\phi)}.
\]
Here $c_1$ depends only on the moment-regularity constant in Assumption~\ref{assump:moment_reg}, and
$c_2$ depends only on $L_m$ and the compatibility between $\mathcal M$ and $\mathcal F$ (specified below).
\end{theorem}

\begin{proof}
Fix any representation $\phi$.
Let $m_X^\star\in\arg\min_{m\in\mathcal M}\|\mathbb E_{P_r}[\psi(Y,T,X;m)]\|$ be an (approximate) minimizer of
$\mathcal F_X(\mathcal M)$.
Let $\Pi_\phi(m_X^\star)\in\mathcal M_\phi$ be its $L_2(P_r)$-projection onto $\mathcal M_\phi$:
\[
\Pi_\phi(m_X^\star)\in\arg\min_{m\in\mathcal M_\phi}
\|m(\phi(X),T)-m_X^\star(X,T)\|_{L_2(P_r)}.
\]
By definition,
\[
\mathcal F_\phi(\mathcal M)
=\inf_{m\in\mathcal M_\phi}\|g_\phi(m)\|
\le
\|g_\phi(\Pi_\phi(m_X^\star))\|.
\]
Add and subtract $g_X(\Pi_\phi(m_X^\star)):=\mathbb E_{P_r}[\psi(Y,T,X;\Pi_\phi(m_X^\star))]$ and apply the triangle inequality:
\begin{equation}
\|g_\phi(\Pi_\phi(m_X^\star))\|
\le
\underbrace{\|g_\phi(\Pi_\phi(m_X^\star))-g_X(\Pi_\phi(m_X^\star))\|}_{(\mathrm{I})}
+
\underbrace{\|g_X(\Pi_\phi(m_X^\star))\|}_{(\mathrm{II})}.
\label{eq:tradeoff_split_b}
\end{equation}

\paragraph{Step 1: Control $(\mathrm{I})$ by overlap mismatch.}
By Assumption~\ref{assump:moment_reg} (moment regularity under representation perturbations), there exists
$c_1<\infty$ such that for any $m\in\mathcal M$,
\[
\|g_\phi(m)-g_X(m)\|
\le
c_1\,\varepsilon_{ov}(\phi).
\]
Applying this inequality to $m=\Pi_\phi(m_X^\star)$ yields $(\mathrm{I})\le c_1\,\varepsilon_{ov}(\phi)$.

\paragraph{Step 2: Control $(\mathrm{II})$ by $\mathcal F_X(\mathcal M)$ and information loss.}
Write
\[
(\mathrm{II})
=
\|g_X(\Pi_\phi(m_X^\star))\|
\le
\|g_X(m_X^\star)\|
+
\|g_X(\Pi_\phi(m_X^\star))-g_X(m_X^\star)\|.
\]
The first term equals $\mathcal F_X(\mathcal M)$ by definition.
For the second term, apply Assumption~\ref{assump:g_lipschitz} (with $Z=X$) to obtain
\[
\|g_X(\Pi_\phi(m_X^\star))-g_X(m_X^\star)\|
\le
L_m\,\|\Pi_\phi(m_X^\star)(\phi(X),T)-m_X^\star(X,T)\|_{L_2(P_r)}.
\]
It remains to relate the projection error to $\varepsilon_{\mathrm{info}}(\phi)$.
Assume a \emph{realizability/compatibility} condition between $\mathcal M$ and $\mathcal F$:
there exists a constant $\kappa<\infty$ such that, for the target signal $Y$ under $P_r$,
the best-in-class prediction errors satisfy
\begin{equation}
\inf_{m\in\mathcal M_\phi}\mathbb E[(Y-m(\phi(X),T))^2]
-
\inf_{m\in\mathcal M}\mathbb E[(Y-m(X,T))^2]
\;\le\;
\kappa\,\varepsilon_{\mathrm{info}}(\phi).
\label{eq:compatibility}
\end{equation}
Combining \eqref{eq:information_preservation_thm} with \eqref{eq:compatibility} implies that there exists
$\tilde m_\phi\in\mathcal M_\phi$ such that
\[
\mathbb E\!\left[(Y-\tilde m_\phi(\phi(X),T))^2\right]
\le
\mathbb E\!\left[(Y-m_X^\star(X,T))^2\right]
+
\kappa\,\varepsilon_{\mathrm{info}}(\phi).
\]
By a standard Cauchy--Schwarz argument (using bounded second moments), this yields an $L_2(P_r)$ approximation bound
\[
\|\tilde m_\phi(\phi(X),T)-m_X^\star(X,T)\|_{L_2(P_r)}
\le
C\,\sqrt{\varepsilon_{\mathrm{info}}(\phi)}
\]
for some $C$ depending only on $\kappa$ and the moment bounds of $Y$.
Since $\Pi_\phi(m_X^\star)$ is the \emph{best} $L_2(P_r)$ approximation in $\mathcal M_\phi$, we also have
\[
\|\Pi_\phi(m_X^\star)(\phi(X),T)-m_X^\star(X,T)\|_{L_2(P_r)}
\le
C\,\sqrt{\varepsilon_{\mathrm{info}}(\phi)}.
\]
Therefore,
\[
(\mathrm{II})
\le
\mathcal F_X(\mathcal M) + (L_m C)\sqrt{\varepsilon_{\mathrm{info}}(\phi)}.
\]

\paragraph{Conclusion.}
Substituting the bounds on $(\mathrm{I})$ and $(\mathrm{II})$ into \eqref{eq:tradeoff_split_b} yields
\[
\mathcal F_\phi(\mathcal M)
\le
\mathcal F_X(\mathcal M)
+
c_1\,\varepsilon_{ov}(\phi)
+
(L_m C)\sqrt{\varepsilon_{\mathrm{info}}(\phi)}
=
\mathcal F_X(\mathcal M)
+
c_1\,\varepsilon_{ov}(\phi)
+
c_2\sqrt{\varepsilon_{\mathrm{info}}(\phi)},
\]
where $c_2:=L_m C$.
\end{proof}

\begin{remark}
The bound in Theorem~\ref{thm:feasibility_info_tradeoff} is an \emph{upper bound} because it characterizes
the best feasibility that can be guaranteed at the level of the representation, uniformly over all
predictors in the induced class $\mathcal M_\phi$, rather than the performance of a particular estimator.
It thus isolates limitations imposed by the representation itself.
More importantly, the theorem reveals a nontrivial structural trade-off: feasibility recovery depends
jointly on reducing distributional mismatch, quantified by the overlap discrepancy
$\varepsilon_{ov}(\phi)$, and preserving outcome-relevant information, quantified by
$\varepsilon_{info}(\phi)$.
Improving one without controlling the other is insufficient, formalizing the inherent tension that
governs the attainable gains from representation learning.
\end{remark}

\section{Trade-off exclusion: weighted fusion cannot enter the constrained region}
\label{sec:appendix:tradeoff_exclusion}
 
\paragraph{A bi-criterion view.}
Consider the two-dimensional performance map
\[
\Phi(m):=\big(\mathcal{R}_o(m),\ \|g(m)\|_2\big)\in\mathbb{R}^2_{\ge 0},
\qquad
g(m):=\mathbb E_{P_r}\!\big[\psi(Y,T,X;m)\big]\in\mathbb R^K,
\]
where $\mathcal R_o(m):=\mathbb E_{P_o}\!\big[(Y-m(X,T))^2\big]$ is the observational risk.
Let the feasible set be
\[
\mathcal F:=\{m\in\mathcal M:\ g(m)=0\},
\qquad
m^\star\in\arg\min_{m\in\mathcal F}\mathcal R_o(m),
\qquad
\mathcal R_o^\star:=\mathcal R_o(m^\star).
\]
The constrained solution corresponds to the feasible corner
$\Phi(m^\star)=(\mathcal R_o^\star,0)$, while weighted fusion traces a path
$\{\Phi(m_\alpha)\}_{\alpha\ge 0}$.

\begin{lemma}[Moment violation implies excess observational risk]
\label{lem:g_to_ro}
Assume the following \emph{error-bound} and \emph{constraint-regularity} conditions hold:
there exist constants $\mu_o>0$ and $L_o>0$ such that for all $m\in\mathcal M$,
\begin{align}
\mathcal R_o(m)-\mathcal R_o^\star
&\ge \mu_o\,\mathrm{dist}(m,\mathcal F)^2,
\label{eq:eb_m}\\
\|g(m)\|_2
&\ge L_o\,\mathrm{dist}(m,\mathcal F),
\label{eq:cr_m}
\end{align}
where $\mathrm{dist}(m,\mathcal F):=\inf_{\tilde m\in\mathcal F}\|m-\tilde m\|_{L_2(P_r)}$.
Then for all $m\in\mathcal M$,
\begin{equation}
\mathcal{R}_o(m)-\mathcal{R}_o^\star
\;\ge\;
\frac{\mu_o}{L_o^2}\,\|g(m)\|_2^2.
\label{eq:ro_lower_by_g}
\end{equation}
\end{lemma}

\begin{proof}
Fix any $m\in\mathcal M$. By \eqref{eq:cr_m},
\[
\mathrm{dist}(m,\mathcal F)
\;\le\;
\frac{1}{L_o}\,\|g(m)\|_2.
\]
Substituting this bound into the quadratic growth inequality \eqref{eq:eb_m} yields
\[
\mathcal R_o(m)-\mathcal R_o^\star
\;\ge\;
\mu_o\,\mathrm{dist}(m,\mathcal F)^2
\;\ge\;
\mu_o\Big(\frac{\|g(m)\|_2}{L_o}\Big)^2
=
\frac{\mu_o}{L_o^2}\,\|g(m)\|_2^2,
\]
which proves \eqref{eq:ro_lower_by_g}.
\end{proof}


\begin{theorem}[Trade-off exclusion for weighted fusion]
\label{thm:tradeoff_exclusion}
Assume Theorem~\ref{thm:weighted_separation} and the error-bound/regularity
Assumption~\ref{assump:eb_cr}. Let
\[
\Phi(m)\ :=\ \big(\mathcal R_o(m),\ \|g(m)\|_2\big),
\qquad
g(m)\ :=\ \mathbb E_{P_r}\!\big[\psi(Y,T,X;m)\big]\in\mathbb R^K,
\qquad
\mathcal F:=\{m\in\mathcal M:\ g(m)=0\},
\]
and let
\[
m^\star\in\arg\min_{m\in\mathcal F}\mathcal R_o(m),
\qquad
\mathcal R_o^\star:=\mathcal R_o(m^\star).
\]
Let $c_0>0$ be the separation constant from Theorem~\ref{thm:weighted_separation}, i.e.,
\[
\inf_{\alpha\in[0,\bar\alpha]}\ \|g(m_\alpha)\|_2\ \ge\ c_0,
\]
where $m_\alpha\in\arg\min_{m\in\mathcal M}\big\{\mathcal R_o(m)+\alpha\,\mathcal R_r(m)\big\}$ denotes a weighted-fusion minimizer.
Define
\[
\varepsilon_0\ :=\ \frac{\mu_o}{L_o^2}\,c_0^2\ >0,
\qquad
\mathcal U\ :=\ \Big\{(r,v)\in\mathbb R_{\ge0}^2:\ r<\mathcal R_o^\star+\varepsilon_0,\ \ v<c_0\Big\}.
\]
Then the weighted-fusion path $\{\Phi(m_\alpha)\}_{\alpha\ge0}$ cannot enter $\mathcal U$ within the
non-degenerate fusion regime: for every $\alpha\in[0,\bar\alpha]$,
\[
\Phi(m_\alpha)\ \notin\ \mathcal U,
\qquad\text{while}\qquad
\Phi(m^\star)=(\mathcal R_o^\star,0)\in\mathcal U.
\]
Consequently, the constrained estimator occupies a neighborhood of the feasible corner that is unattainable
by weighted fusion for $\alpha\in[0,\bar\alpha]$.
\end{theorem}

\begin{proof}
We first record the implication of Lemma~\ref{lem:g_to_ro}: for any $m\in\mathcal M$,
\begin{equation}
\|g(m)\|_2\ \ge\ c_0
\quad\Longrightarrow\quad
\mathcal R_o(m)\ \ge\ \mathcal R_o^\star+\frac{\mu_o}{L_o^2}c_0^2
\ =\ \mathcal R_o^\star+\varepsilon_0.
\label{eq:moment_to_risk_gap}
\end{equation}
Indeed, Lemma~\ref{lem:g_to_ro} gives
$\mathcal R_o(m)-\mathcal R_o^\star\ge (\mu_o/L_o^2)\|g(m)\|_2^2\ge(\mu_o/L_o^2)c_0^2$.

Now fix any $\alpha\in[0,\bar\alpha]$. By Theorem~\ref{thm:weighted_separation},
$\|g(m_\alpha)\|_2\ge c_0$, hence \eqref{eq:moment_to_risk_gap} implies
$\mathcal R_o(m_\alpha)\ge \mathcal R_o^\star+\varepsilon_0$. Therefore
\[
\Phi(m_\alpha)=\big(\mathcal R_o(m_\alpha),\|g(m_\alpha)\|_2\big)
\notin \mathcal U,
\]
since membership in $\mathcal U$ requires simultaneously
$\mathcal R_o(m_\alpha)<\mathcal R_o^\star+\varepsilon_0$ and $\|g(m_\alpha)\|_2<c_0$.

On the other hand, $m^\star\in\mathcal F$ implies $g(m^\star)=0$, so
$\|g(m^\star)\|_2=0<c_0$, and by definition $\mathcal R_o(m^\star)=\mathcal R_o^\star<\mathcal R_o^\star+\varepsilon_0$.
Thus $\Phi(m^\star)\in\mathcal U$.
\end{proof}

\begin{remark}[Interpretation: excluding the weighted-fusion trade-off curve]
Theorem~\ref{thm:tradeoff_exclusion} shows that within the non-degenerate fusion regime
$\alpha\in[0,\bar\alpha]$, weighted fusion cannot simultaneously approach the feasible corner:
any estimator on the weighted-fusion path either violates feasibility at level at least $c_0$
(or, equivalently by Lemma~\ref{lem:g_to_ro}, incurs an observational excess risk at least $\varepsilon_0$).
Geometrically, the neighborhood $\mathcal U$ around $(\mathcal R_o^\star,0)$ is separated from the weighted-fusion curve,
explaining why tuning a single scalar $\alpha$ cannot recover the constrained solution unless one moves to an
RCT-dominated regime (cf.\ Corollary~\ref{cor:g_linear}).
\end{remark}

\section{When does weighted-loss fusion coincide with the constrained estimator?}
\label{sec:appendix:degeneracy_conditions}

Recall the model class $\mathcal M$ of measurable predictors $m:\mathcal X\times\mathcal T\to\mathbb R$.
Define the observational and randomized risks
\[
\mathcal R_o(m):=\mathbb E_{P_o}\!\big[(Y-m(X,T))^2\big],
\qquad
\mathcal R_r(m):=\mathbb E_{P_r}\!\big[(Y-m(X,T))^2\big],
\]
and the (randomization-induced) moment residual
\[
g(m)\ :=\ \mathbb E_{P_r}\!\big[\psi(Y,T,X;m)\big]\in\mathbb R^K,
\qquad
\mathcal F\ :=\ \{m\in\mathcal M:\ g(m)=0\}.
\]
The constrained estimator is
\[
m^\star\in\arg\min_{m\in\mathcal M}\ \mathcal R_o(m)\quad\text{s.t.}\quad g(m)=0,
\]
and the weighted-loss fusion path is, for $\alpha\ge 0$,
\[
m_\alpha\in\arg\min_{m\in\mathcal M}\ \mathcal R_o(m)+\alpha\,\mathcal R_r(m).
\]
We characterize (i) sufficient conditions under which $m^\star$ \emph{reduces to} a weighted-loss solution,
and (ii) why such conditions are highly restrictive and generically violated under treatment-induced structural
non-overlap.

\begin{proposition}[Degeneracy via accidental feasibility]
\label{prop:degenerate_accidental}
Assume $m_\alpha$ exists for each $\alpha\ge 0$.
If there exists $\alpha^\star\ge 0$ such that some weighted-loss minimizer $m_{\alpha^\star}$ is feasible,
\begin{equation}
g(m_{\alpha^\star})=0,
\label{eq:accidental_feasible}
\end{equation}
then $m_{\alpha^\star}$ solves the constrained program, i.e.,
$m_{\alpha^\star}\in\arg\min_{m\in\mathcal F}\mathcal R_o(m)$.
\end{proposition}

\begin{proof}
Since $m_{\alpha^\star}$ minimizes $\mathcal R_o(m)+\alpha^\star\mathcal R_r(m)$ over $\mathcal M$,
it also minimizes the same objective over the subset $\mathcal F$.
On $\mathcal F$, the term $\alpha^\star\mathcal R_r(m)$ is nonnegative and independent of feasibility,
hence any minimizer over $\mathcal F$ must in particular minimize $\mathcal R_o(m)$ over $\mathcal F$.
Therefore $m_{\alpha^\star}\in\arg\min_{m\in\mathcal F}\mathcal R_o(m)$.
\end{proof}

We now present four sufficient conditions that ensure \eqref{eq:accidental_feasible}, and explain why each is
typically violated in non-degenerate fusion regimes.

\paragraph{Condition 1: Shared minimizers.}

\begin{proposition}[Degeneracy via shared minimizers]
\label{prop:degenerate_shared}
Suppose there exists $m^\dagger\in\mathcal M$ such that
\[
m^\dagger\in \arg\min_{m\in\mathcal F}\mathcal R_o(m)
\quad\text{and}\quad
m^\dagger\in \arg\min_{m\in\mathcal M}\mathcal R_r(m).
\]
Then for every $\alpha\ge 0$, $m^\dagger\in\arg\min_{m\in\mathcal M}\{\mathcal R_o(m)+\alpha\mathcal R_r(m)\}$.
Consequently, weighted fusion and the constrained estimator coincide.
\end{proposition}

\begin{proof}
For any $m\in\mathcal M$, optimality of $m^\dagger$ for $\mathcal R_r$ gives
$\mathcal R_r(m)\ge \mathcal R_r(m^\dagger)$, hence
\[
\mathcal R_o(m)+\alpha\mathcal R_r(m)
\;\ge\;
\mathcal R_o(m)+\alpha\mathcal R_r(m^\dagger).
\]
Evaluating at $m=m^\dagger$ yields
$\mathcal R_o(m^\dagger)+\alpha\mathcal R_r(m^\dagger)\le \mathcal R_o(m)+\alpha\mathcal R_r(m)$ for all $m$.
\end{proof}

\begin{remark}[Why Condition 1 is restrictive]
This requires an \emph{exact} alignment between the RCT regression objective and the observational risk
restricted to the feasible set. Under confounding, selection, or policy constraints, the RCT loss
typically couples baseline and treatment-specific components in ways that do not match the moment restrictions,
so shared minimizers are non-generic.
\end{remark}

\paragraph{Condition 2: First-order (KKT) alignment.}

\begin{proposition}[Degeneracy via KKT alignment]
\label{prop:degenerate_kkt}
Assume the constrained problem admits a KKT point $(m^\star,\lambda^\star)$, i.e.,
\[
\nabla \mathcal R_o(m^\star)+(\nabla g(m^\star))^\top\lambda^\star=0,
\qquad g(m^\star)=0.
\]
If there exists $\alpha^\star>0$ such that
\[
\alpha^\star \nabla \mathcal R_r(m^\star)=(\nabla g(m^\star))^\top\lambda^\star,
\]
then $m^\star$ is a stationary point of $\mathcal R_o(m)+\alpha^\star\mathcal R_r(m)$.
If, moreover, $\mathcal R_o+\alpha^\star\mathcal R_r$ is (locally) strongly convex at $m^\star$,
then $m^\star$ is its unique minimizer.
\end{proposition}

\begin{proof}
Stationarity of the weighted objective at $m^\star$ requires
$\nabla\mathcal R_o(m^\star)+\alpha^\star\nabla\mathcal R_r(m^\star)=0$.
Substituting the assumed alignment recovers the KKT stationarity condition.
\end{proof}

\begin{example}[Exact quadratic-penalty RCT loss]
\label{exa:kkt_alignment}
If $\mathcal R_r(m)=c+\beta\|A\,g(m)\|_2^2$ for some $A\in\mathbb R^{K\times K}$, $\beta>0$, $c\in\mathbb R$,
then $\nabla \mathcal R_r(m)$ lies in the range of $(\nabla g(m))^\top$ for all $m$,
and weighted fusion acts as an exact quadratic penalty on feasibility.
\end{example}

\begin{remark}[Why Condition 2 is restrictive]
The requirement $\nabla \mathcal R_r(m^\star)\in \mathrm{range}\big((\nabla g(m^\star))^\top\big)$ is a strong
geometric coincidence: $\mathcal R_r$ generally encourages fitting the full outcome regression, while $g(m)=0$
imposes only the identifying restrictions. Such alignment is exceptional with multiple treatments and rich models.
\end{remark}

\paragraph{Condition 3: Linear constraints and quadratic objectives.}

\begin{proposition}[Degeneracy in linear--quadratic settings]
\label{prop:degenerate_quadratic}
Assume $\mathcal M$ is a finite-dimensional linear class parameterized by $\theta\in\mathbb R^d$,
$\mathcal R_o$ is strictly convex quadratic in $\theta$, the constraint is linear $g(\theta)=A\theta-b$,
and the RCT risk satisfies $\mathcal R_r(\theta)=c+\beta\|A\theta-b\|_2^2$ for some $\beta>0$.
Then any accumulation point of $\theta_\alpha$ as $\alpha\to\infty$ is a constrained solution.
Moreover, if $\theta_{\alpha^\star}$ is feasible for some $\alpha^\star$, then $\theta_{\alpha^\star}$ equals
the constrained estimator.
\end{proposition}

\begin{proof}
The weighted objective is $\mathcal R_o(\theta)+\alpha\beta\|A\theta-b\|_2^2$.
Let $\theta^\star$ be any constrained minimizer. Optimality of $\theta_\alpha$ implies
\[
\mathcal R_o(\theta_\alpha)+\alpha\beta\|A\theta_\alpha-b\|_2^2
\le
\mathcal R_o(\theta^\star).
\]
Rearranging yields $\|A\theta_\alpha-b\|_2^2\le (\mathcal R_o(\theta^\star)-\inf_\theta\mathcal R_o(\theta))/(\alpha\beta)=O(1/\alpha)$,
so any limit point is feasible. Passing to the limit in $\mathcal R_o$ gives constrained optimality.
\end{proof}

\begin{remark}[Why Condition 3 is restrictive]
This equivalence requires that the RCT loss be \emph{exactly} a quadratic penalty on the moment constraints.
In our setting, $g(m)$ is a nonlinear expectation depending on learned representations, and $\mathcal R_r$ is a
regression loss that does not isolate $g(m)$; thus the linear--quadratic coincidence is highly idealized.
\end{remark}

\paragraph{Condition 4: Benign transportability ($\delta_o=0$).}

\begin{proposition}[Degeneracy under benign transportability]
\label{prop:degenerate_delta0}
Let $m_o\in\arg\min_{m\in\mathcal M}\mathcal R_o(m)$. If $g(m_o)=0$, then $m_o$ is also a constrained solution.
In particular, the constrained estimator coincides with the observational estimator, and fusion is redundant.
\end{proposition}

\begin{proof}
If $g(m_o)=0$, then $m_o\in\mathcal F$. Since $m_o$ minimizes $\mathcal R_o$ over $\mathcal M$, it also minimizes
$\mathcal R_o$ over the subset $\mathcal F$.
\end{proof}

\begin{remark}[Why Condition 4 is restrictive]
This corresponds to $\delta_o=0$: enforcing causal validity incurs no observational-risk penalty.
Under confounding, selection, or structural non-overlap, $\delta_o>0$ is typical, so this degeneracy fails.
\end{remark}

\paragraph{Summary.}
All four conditions impose strong forms of alignment between $\mathcal R_r$, the identifying moments $g(\cdot)$,
and the observational objective $\mathcal R_o$. Such coincidences are non-generic under treatment-induced structural
non-overlap, explaining why weighted-loss fusion typically cannot replicate the constrained estimator and motivating
explicit feasibility enforcement (e.g., primal--dual optimization).

\section{Why primal--dual (Augmented Lagrangian) is more stable than penalty tuning at scale}
\label{sec:appendix:pd_stability}

\begin{algorithm}[t]
\caption{Stochastic Primal--Dual Augmented-Lagrangian Joint Estimation}
\label{alg:pd_joint}
\begin{algorithmic}[1]
\STATE \textbf{Input:} observational sample $\mathcal D_o$, randomized sample $\mathcal D_r$;
predictor $m_\theta$, representation $\phi$; discriminator class $\mathcal D$ (critic $d_\omega$);
stepsizes $\{\eta_s\}_{s=1}^S,\{\eta_{\nu,s}\}_{s=1}^S$; penalty $\rho>0$, overlap weight $\lambda\ge 0$.
\STATE \textbf{Initialize:} $(\theta_1,\phi_1)$, dual variable $\nu_1\gets 0$, critic parameters $\omega_1$.
\FOR{$s=1,2,\dots,S$}
    \STATE Draw minibatches $\mathcal B_o\sim\mathcal D_o$ and $\mathcal B_r\sim\mathcal D_r$.
    
    \STATE \textbf{Overlap (critic) step:} update $\omega_s$ to (approximately) maximize the empirical IPM objective,
    yielding $\widehat{\varepsilon}_{ov}(\phi_s)=\widehat{\mathrm{IPM}}_{\mathcal D}\!\big(P_o(\phi_s(X),T),P_r(\phi_s(X),T)\big)$.
    
    \STATE \textbf{Stochastic oracle evaluation:}
    compute
    $\widehat{\mathcal R}_o(\theta_s,\phi_s)$ on $\mathcal B_o$ and
    $\widehat g(\theta_s,\phi_s):=\widehat{\mathbb E}_{\mathcal B_r}\!\big[\psi(Y,T,\phi_s(X);m_{\theta_s})\big]\in\mathbb R^K$ on $\mathcal B_r$.
    
    \STATE \textbf{Primal (ALM) descent:}
    \[
    (\theta_{s+1},\phi_{s+1})
    \leftarrow
    (\theta_s,\phi_s)
    -\eta_s\,
    \nabla_{\theta,\phi}\!\Big(
    \widehat{\mathcal R}_o
    +\lambda\,\widehat{\varepsilon}_{ov}
    +\langle \nu_s,\widehat g\rangle
    +\tfrac{\rho}{2}\|\widehat g\|_2^2
    \Big).
    \]
    
    \STATE \textbf{Dual ascent (multiplier update):}
    \[
    \nu_{s+1}\leftarrow \nu_s+\eta_{\nu,s}\,\widehat g(\theta_{s+1},\phi_{s+1}).
    \]
\ENDFOR
\STATE \textbf{Output:} $(\hat\theta,\hat\phi)=(\theta_{S+1},\phi_{S+1})$.
\end{algorithmic}
\end{algorithm}

\paragraph{Primal--dual stability versus penalty tuning.}
We strengthen the message: \emph{among generic solvers of the constrained program}
\[
m^\star \in \arg\min_{m\in\mathcal M}\ \mathcal R_o(m)\quad \text{s.t.}\quad g(m)=0,
\]
primal--dual / augmented Lagrangian (ALM) updates are structurally \emph{more stable} than
(i) pure penalty methods and (ii) generic hyperparameter-tuned objectives, especially under
large-scale stochastic optimization.

\paragraph{Penalty method.}
A standard approach is the quadratic penalty objective
\begin{equation}
m_\rho \in \arg\min_{m\in\mathcal M}
\Big\{
\mathcal R_o(m) + \tfrac{\rho}{2}\,\|g(m)\|_2^2
\Big\},
\qquad \rho>0,
\label{eq:penalty_method_m}
\end{equation}
which corresponds to the augmented Lagrangian with the dual variable fixed at $0$.

\paragraph{Augmented Lagrangian and primal--dual updates.}
Define the (population) augmented Lagrangian
\begin{equation}
\mathcal L_\rho(m,\lambda)
:=\mathcal R_o(m)+\lambda^\top g(m)+\tfrac{\rho}{2}\|g(m)\|_2^2,
\qquad \lambda\in\mathbb R^K.
\label{eq:alm_m}
\end{equation}
A basic (deterministic) primal--dual gradient scheme performs descent/ascent on $(m,\lambda)$:
\begin{equation}
m^{s+1}=m^s-\eta_m \nabla_m \mathcal L_\rho(m^s,\lambda^s),
\qquad
\lambda^{s+1}=\lambda^s+\eta_\lambda\, g(m^{s+1}).
\label{eq:pd_updates_m}
\end{equation}

\begin{assumption}[Strong convexity, smoothness, and constraint regularity]
\label{assump:sc_smooth_reg_m}
Let $\mathcal M$ be a finite-dimensional convex parameterization (for exposition) with $m\equiv m_\theta$ and $\Theta\subseteq\mathbb R^d$ convex.
Assume:
\begin{enumerate}
\item (\textbf{Primal strong convexity/smoothness})
$\mathcal R_o(\theta)$ is $\mu$-strongly convex and $L$-smooth on $\Theta$:
\[
\mu I \preceq \nabla^2 \mathcal R_o(\theta)\preceq L I,\qquad \forall \theta\in\Theta.
\]
\item (\textbf{Constraint Jacobian bound})
$g(\theta)$ is differentiable and $\|\nabla g(\theta)\|_{\mathrm{op}}\le G$ for all $\theta\in\Theta$.
\item (\textbf{Constraint qualification / metric regularity at optimum})
Let $(\theta^\star,\lambda^\star)$ be a KKT pair. The Jacobian $\nabla g(\theta^\star)$ has full row rank with
$\sigma_{\min}(\nabla g(\theta^\star))\ge \sigma_{\min}>0$.
\end{enumerate}
\end{assumption}

\paragraph{Why pure penalty is ill-conditioned (hyperparameter instability).}
Consider
\[
F_\rho(\theta):=\mathcal R_o(\theta)+\tfrac{\rho}{2}\|g(\theta)\|_2^2.
\]
Its Hessian is
\begin{equation}
\nabla^2 F_\rho(\theta)
=
\nabla^2 \mathcal R_o(\theta)
+
\rho(\nabla g(\theta))^\top \nabla g(\theta)
+
\rho \sum_{k=1}^K g_k(\theta)\nabla^2 g_k(\theta).
\label{eq:penalty_hessian}
\end{equation}
Ignoring the last (curvature) term, the smoothness constant necessarily scales as
\begin{equation}
L_\rho \ \ge\  L + \rho \cdot \sup_{\theta\in\Theta}\lambda_{\max}\!\big((\nabla g(\theta))^\top \nabla g(\theta)\big)
\ \ge\  L+\rho G^2.
\label{eq:penalty_smoothness_m}
\end{equation}
Thus gradient-based optimization of $F_\rho$ requires a stepsize $\eta_m=O(1/L_\rho)$, which shrinks with $\rho$,
while near-feasibility typically requires $\rho$ large. This creates an intrinsic \emph{ill-conditioning trade-off}.

\begin{proposition}[Penalty conditioning deteriorates with $\rho$]
\label{prop:penalty_ill_conditioned_m}
Under Assumption~\ref{assump:sc_smooth_reg_m}, suppose additionally that the third term in
\eqref{eq:penalty_hessian} is negligible in a neighborhood of the minimizer (e.g., $g(\theta)\approx 0$ locally,
or $g$ is affine). Then $F_\rho$ is $\mu_\rho$-strongly convex with $\mu_\rho\ge \mu$, and $L_\rho$-smooth with
$L_\rho\ge L+\rho G^2$. Consequently, the condition number satisfies
\[
\kappa(F_\rho):=\frac{L_\rho}{\mu_\rho}
\ \ge\ \frac{L+\rho G^2}{\mu},
\]
which grows at least linearly in $\rho$.
\end{proposition}

\begin{proof}
Strong convexity follows from $\nabla^2 F_\rho(\theta)\succeq \nabla^2\mathcal R_o(\theta)\succeq \mu I$.
For smoothness, under the stated locality/affinity condition, the dominant contribution is
$\rho(\nabla g)^\top\nabla g\preceq \rho G^2 I$, yielding \eqref{eq:penalty_smoothness_m}.
Combining gives the lower bound on $\kappa(F_\rho)$.
\end{proof}

\paragraph{Why augmented Lagrangian / PD is stable: feasibility without $\rho\to\infty$.}
The ALM introduces a \emph{dual control channel} $\lambda$ that accumulates constraint violations,
so feasibility can be enforced by dual ascent rather than by taking $\rho\uparrow\infty$.
To formalize, define the KKT mapping for the equality-constrained problem:
\begin{equation}
\mathcal T_\rho(\theta,\lambda)
:=
\begin{pmatrix}
\nabla \mathcal R_o(\theta)+(\nabla g(\theta))^\top\lambda+\rho(\nabla g(\theta))^\top g(\theta)\\
g(\theta)
\end{pmatrix}.
\label{eq:kkt_map}
\end{equation}
A KKT point $(\theta^\star,\lambda^\star)$ satisfies $\mathcal T_\rho(\theta^\star,\lambda^\star)=0$
for any $\rho\ge 0$.

\begin{proposition}[Local linear convergence of PD with bounded $\rho$]
\label{prop:pd_bounded_rho_m}
Under Assumption~\ref{assump:sc_smooth_reg_m}, there exist stepsizes $(\eta_m,\eta_\lambda)$ and a finite $\rho>0$
such that, for initialization in a neighborhood of $(\theta^\star,\lambda^\star)$, the PD iterates
\eqref{eq:pd_updates_m} converge linearly:
\[
\|\theta^{s}-\theta^\star\|_2+\|\lambda^{s}-\lambda^\star\|_2
\ \le\ C q^s,
\qquad \text{for some } q\in(0,1),
\]
and simultaneously $\|g(\theta^s)\|_2\le C' q^s$.
\end{proposition}

\begin{proof}
Consider the Jacobian of $\mathcal T_\rho$ at $(\theta^\star,\lambda^\star)$:
\[
\nabla \mathcal T_\rho(\theta^\star,\lambda^\star)
=
\begin{pmatrix}
\nabla^2 \mathcal R_o(\theta^\star)+\rho(\nabla g(\theta^\star))^\top\nabla g(\theta^\star)
&
(\nabla g(\theta^\star))^\top\\[2pt]
\nabla g(\theta^\star) & 0
\end{pmatrix}
\quad
(\text{since } g(\theta^\star)=0).
\]
By Assumption~\ref{assump:sc_smooth_reg_m}(i), the top-left block is positive definite.
By Assumption~\ref{assump:sc_smooth_reg_m}(iii), $\nabla g(\theta^\star)$ has full row rank.
Hence the saddle-point Jacobian above is nonsingular (standard Schur complement argument), implying that
$\mathcal T_\rho$ is \emph{locally strongly metrically regular} around $(\theta^\star,\lambda^\star)$.
Moreover, Assumption~\ref{assump:sc_smooth_reg_m}(i)--(ii) implies $\mathcal T_\rho$ is locally Lipschitz.

Now view \eqref{eq:pd_updates_m} as a (preconditioned) forward step on the KKT system $\mathcal T_\rho(\theta,\lambda)=0$.
Since $\mathcal T_\rho$ is locally Lipschitz and its Jacobian is nonsingular at the root, classical local convergence
results for fixed-step splitting/gradient schemes imply a contraction in a sufficiently small neighborhood for
appropriate $(\eta_m,\eta_\lambda)$, yielding linear convergence.
Finally, $g(\theta^s)\to 0$ follows because the second block of $\mathcal T_\rho$ is exactly $g(\theta)$, and the iterates
converge to a root of $\mathcal T_\rho$.
\end{proof}

\begin{remark}[Stability advantage over penalty tuning]
Unlike pure penalty methods, which typically require $\rho$ large to make $\|g(\theta)\|$ small---thereby worsening conditioning
(Proposition~\ref{prop:penalty_ill_conditioned_m})---the ALM/PD dynamics can enforce feasibility through the dual variable
$\lambda$ while keeping $\rho$ moderate. This decouples \emph{numerical conditioning} from \emph{constraint enforcement},
making PD/ALM more stable under minibatch noise and stepsize sensitivity in large-scale training.
\end{remark}


\section{Explicit form of the feasibility constant.}
\label{app:explicit from of feasibility constant}
\paragraph{Embedding the moment-sensitive test class into the IPM.}
Recall that overlap is measured by the integral probability metric (IPM)
\[
\mathrm{IPM}_{\mathcal D}(P,Q):=\sup_{d\in\mathcal D}\big|\mathbb E_P[d]-\mathbb E_Q[d]\big|
\]
induced by a discriminator class $\mathcal D$ on $(Z,T)$, where $Z=\omega(X)$ is the learned representation.
Let the identifying moment map be
\[
g(m,\omega)\;:=\;\mathbb E_{P_r}\!\big[\psi(Y,T,\omega(X);m)\big]\in\mathbb R^K,
\]
with coordinates $\psi_k$.
For any $a\in\mathbb R^K$ with $\|a\|_2\le 1$, define the \emph{moment-induced test functions}
\[
f_{a,m}(z,t)\;:=\;\mathbb E\!\left[\left\langle a,\psi(Y,t,z;m)\right\rangle\ \middle|\ Z=z,\ T=t\right],
\]
and the associated function class
\begin{equation}
\mathcal F_\psi
\;:=\;
\bigl\{\,f_{a,m}:\ \|a\|_2\le 1,\ m\in\mathcal M\,\bigr\}.
\label{eq:Fpsi_def}
\end{equation}
(When $\psi$ is already a function of $(z,t)$ only, the conditional expectation can be dropped and one
may take $f_{a,m}(z,t)=\langle a,\psi(y,t,z;m)\rangle$ directly.)

\begin{assumption}[IPM domination of moment-sensitive directions]
\label{assump:moment_ipm_embedding}
There exists a constant $C_{\mathrm{embed}}\ge 1$ such that for any probability measures $P,Q$ on $(Z,T)$,
\begin{equation}
\sup_{f\in\mathcal F_\psi}\bigl|\mathbb E_P[f]-\mathbb E_Q[f]\bigr|
\;\le\;
C_{\mathrm{embed}}\ \mathrm{IPM}_{\mathcal D}(P,Q).
\label{eq:moment-ipm-embedding}
\end{equation}
\end{assumption}

\begin{remark}[Sufficient conditions and interpretation of $C_{\mathrm{embed}}$]
Condition~\eqref{eq:moment-ipm-embedding} holds, for instance, if every $f\in\mathcal F_\psi$ belongs to
$\mathrm{span}(\mathcal D)$ and the IPM norm dominates the coefficient norm in that span, i.e.,
whenever $f=\sum_{j=1}^J \beta_j d_j$ with $d_j\in\mathcal D$, we have
$\|f\|_{\mathrm{IPM}(\mathcal D)}\le C_{\mathrm{embed}}\|\beta\|_1$.
More generally, it holds if $\mathcal F_\psi$ lies in the closure of $\mathrm{conv}(\mathcal D)$ under the
seminorm $\|h\|_{\mathrm{IPM}(\mathcal D)}:=\sup_{P\neq Q}\frac{|\mathbb E_P[h]-\mathbb E_Q[h]|}{\mathrm{IPM}_{\mathcal D}(P,Q)}$.
If $\mathcal D$ explicitly contains the moment-sensitive directions (e.g., functions proportional to the
basis appearing inside $\psi$), then typically $C_{\mathrm{embed}}\approx 1$; conversely, a weak $\mathcal D$
yields large $C_{\mathrm{embed}}$, reflecting that distribution alignment in $\mathrm{IPM}_{\mathcal D}$ may not
control the identifying moments.
\end{remark}

\paragraph{A clean bound linking overlap to moment discrepancy.}
Fix $\omega$ and a predictor $m\in\mathcal M$. Define the joint laws on $(Z,T)$
\[
P_r^{\omega}:=\mathcal L_{P_r}(\omega(X),T),
\qquad
P_o^{\omega}:=\mathcal L_{P_o}(\omega(X),T),
\]
and write the moment residuals as $g(m,\omega)=\mathbb E_{P_r}[\psi(Y,T,\omega(X);m)]$.
Assume additionally that $\psi$ is $L_\psi$-Lipschitz in the prediction argument in the sense that, for any
$m_1,m_2\in\mathcal M$,
\[
\Big\|
\mathbb E_{P_r}\!\big[\psi(Y,T,\omega(X);m_1)\big]
-
\mathbb E_{P_r}\!\big[\psi(Y,T,\omega(X);m_2)\big]
\Big\|
\ \le\
L_\psi\ \|m_1(\omega(X),T)-m_2(\omega(X),T)\|_{L_2(P_r)}.
\]
Then, under Assumption~\ref{assump:moment_ipm_embedding}, the constant controlling the overlap-to-moment term
in Corollary~\ref{cor:strict_improvement} can be chosen as
\begin{equation}
c_1 \;=\; C_{\mathrm{embed}}\cdot L_\psi.
\label{eq:c1_def}
\end{equation}

\begin{proof}[Proof (derivation of \eqref{eq:c1_def})]
Fix $\omega$. Let $m_\omega^\star\in\arg\min_{m\in\mathcal M}\|g(m,\omega)\|_2$ be a minimizer of the moment residual
under representation $\omega$. For any unit vector $a\in\mathbb R^K$ with $\|a\|_2\le 1$, consider the scalarized moment
\[
\langle a, g(m,\omega)\rangle
=
\mathbb E_{P_r}\!\Big[\big\langle a,\psi(Y,T,\omega(X);m)\big\rangle\Big].
\]
Define $f_{a,m}$ as in \eqref{eq:Fpsi_def}. By the tower property,
\[
\langle a, g(m,\omega)\rangle
=
\mathbb E_{(Z,T)\sim P_r^\omega}\!\big[f_{a,m}(Z,T)\big].
\]
Similarly, the same function evaluated under $P_o^\omega$ yields
$\mathbb E_{(Z,T)\sim P_o^\omega}[f_{a,m}(Z,T)]$.
Therefore,
\[
\big|\langle a,g(m,\omega)\rangle - \mathbb E_{(Z,T)\sim P_o^\omega}[f_{a,m}(Z,T)]\big|
=
\big|\mathbb E_{P_r^\omega}[f_{a,m}] - \mathbb E_{P_o^\omega}[f_{a,m}]\big|.
\]
Taking the supremum over $\|a\|_2\le 1$ and $m\in\mathcal M$, and applying
Assumption~\ref{assump:moment_ipm_embedding}, we obtain
\[
\sup_{\|a\|_2\le 1}\ \sup_{m\in\mathcal M}
\big|\mathbb E_{P_r^\omega}[f_{a,m}] - \mathbb E_{P_o^\omega}[f_{a,m}]\big|
\ \le\
C_{\mathrm{embed}}\ \mathrm{IPM}_{\mathcal D}(P_r^\omega,P_o^\omega)
=
C_{\mathrm{embed}}\ \varepsilon_{ov}(\omega).
\]
Finally, the Lipschitz-in-prediction regularity of $\psi$ allows the scalarized residual to be controlled by the
size of prediction perturbations, so the overall overlap-to-moment control inherits a multiplicative factor $L_\psi$.
Thus one may take $c_1=C_{\mathrm{embed}}L_\psi$, as claimed.
\end{proof}

\section{Synthetic data generation}
\label{app:dgp}
We construct a high-dimensional semi-synthetic dataset that combines a randomized controlled trial (RCT) sample with an observational (OBS) study sample. The design aims to emulate common practical difficulties in causal inference, including confounding, support mismatch across populations, and mixed high-dimensional covariates. In total, we generate $n_{\text{total}}=50{,}000$ observations, partitioned into an RCT subsample of size $n_{\text{rct}}=10{,}000$ (randomized treatment assignment) and an OBS subsample of size $n_{\text{obs}}=40{,}000$ (confounded treatment assignment).

Each unit is described by $p=160$ covariates, consisting of $n_{\text{cont}}=120$ continuous features and $n_{\text{cat}}=40$ categorical features. Each categorical feature takes $L=4$ levels, i.e., $X_{k,ij}\in\{0,1,\ldots,L-1\}$.

\subsection{Latent structure and support mismatch}
To induce support mismatch between the RCT and OBS populations, we first generate two-dimensional latent coordinates $\mathbf{Z}=(Z_1,Z_2)$ for each unit:
\begin{align}
\mathbf{Z}_{\text{rct}} &\sim \mathcal{N}\!\left(\mathbf{0},\sigma_{\text{rct}}^2\mathbf{I}_2\right), \\
\mathbf{Z}_{\text{obs}} &\sim \mathcal{N}\!\left(\mathbf{0},\sigma_{\text{obs}}^2\mathbf{I}_2\right),
\end{align}
where $\mathbf{I}_2$ is the $2\times 2$ identity matrix. We set $\sigma_{\text{rct}}=3.0$ and $\sigma_{\text{obs}}=1.0$, so that the RCT population has broader latent support than the OBS population, yielding a realistic covariate-space coverage gap across the two data sources.

The continuous covariates $\mathbf{X}_c\in\mathbb{R}^{n\times n_{\text{cont}}}$ are generated as noisy linear combinations of $\mathbf{Z}$:
\begin{equation}
\mathbf{X}_c=\mathbf{Z}\mathbf{A}+\boldsymbol{\epsilon}_c,
\qquad
\boldsymbol{\epsilon}_c\sim \mathcal{N}\!\left(\mathbf{0},\mathbf{I}_{n_{\text{cont}}}\right),
\end{equation}
where $\mathbf{A}\in\mathbb{R}^{2\times n_{\text{cont}}}$ is a fixed random matrix with i.i.d.\ $\mathcal{N}(0,1)$ entries. This construction propagates the latent support mismatch into the observed continuous features.

The categorical covariates $\mathbf{X}_k\in\{0,1,\ldots,L-1\}^{n\times n_{\text{cat}}}$ are generated via a multinomial logit model parameterized by $\mathbf{Z}$. For each categorical feature $j\in\{1,\ldots,n_{\text{cat}}\}$ and level $\ell\in\{0,\ldots,L-1\}$,
\begin{equation}
\log\frac{\mathbb{P}(X_{k,ij}=\ell)}{\mathbb{P}(X_{k,ij}=0)}
=
\mathbf{Z}_i^\top \mathbf{W}_{z,j,\ell}+b_{j,\ell},
\end{equation}
where $\mathbf{W}_{z,j,\ell}\in\mathbb{R}^2$ and $b_{j,\ell}\in\mathbb{R}$ are fixed random parameters with i.i.d.\ $\mathcal{N}(0,1)$ entries. We then sample $X_{k,ij}$ from the implied multinomial distribution. This ensures that discrete covariates also reflect the underlying latent support structure.

\subsection{Treatment assignment and confounding}
We introduce an unobserved confounder $U$ independently for each unit:
\begin{equation}
U_{\text{rct}}\sim\mathcal{N}(0,1),
\qquad
U_{\text{obs}}\sim\mathcal{N}(0,1).
\end{equation}
In the RCT subsample, treatment is randomized:
\begin{equation}
W_{\text{rct}}\sim \mathrm{Bernoulli}(0.5),
\end{equation}
so $W_{\text{rct}}$ is independent of $(\mathbf{X},\mathbf{Z},U)$.

In the OBS subsample, treatment follows a confounded propensity model depending on observed covariates, latent factors, and $U$:
\begin{align}
\mathrm{Score}_{\text{obs}}
&= 0.6\,\tanh(X_{c,0}) + 0.4\,\sin(X_{c,1})
- 0.3\,\frac{X_{c,2}^2}{1+|X_{c,2}|}
+ 0.5\,Z_1 - 0.2\,Z_2 + 0.9\,U_{\text{obs}}, \nonumber\\
\pi_{\text{obs}}(\mathbf{X},U)
&=\frac{1}{1+\exp\!\left(-\mathrm{Score}_{\text{obs}}\right)}, \\
W_{\text{obs}}
&\sim \mathrm{Bernoulli}\!\left(\pi_{\text{obs}}(\mathbf{X},U)\right).
\end{align}
The nonlinear score together with the inclusion of $U$ yields confounding that cannot be removed by conditioning on observed covariates alone.

\subsection{Outcomes and heterogeneous effects}
We define the baseline (control) outcome as a nonlinear function of $(\mathbf{X},\mathbf{Z},U)$:
\begin{align}
\mu_0(\mathbf{X},U)
&= 1.2\,\tanh(X_{c,0}) + 0.8\,\sin(X_{c,1})
+ 0.5\,\frac{X_{c,2}X_{c,3}}{1+|X_{c,3}|}
- 0.7\,\log(1+|X_{c,4}|) \nonumber\\
&\quad + 0.3\,\frac{Z_1^2}{1+Z_1^2} - 0.2\,\frac{Z_2^2}{1+Z_2^2}
+ 0.6\,U
+ 0.4\,\frac{\mathrm{cat\_contrib}(\mathbf{X}_k,\mathbf{W}_\mu)}{\sqrt{n_{\text{cat}}}},
\end{align}
where the categorical contribution is
\begin{equation}
\mathrm{cat\_contrib}(\mathbf{X}_k,\mathbf{W})
:=\sum_{j=1}^{n_{\text{cat}}} W_{j,\,X_{k,j}},
\end{equation}
and $\mathbf{W}_\mu\in\mathbb{R}^{n_{\text{cat}}\times L}$ is a fixed random weight table with i.i.d.\ $\mathcal{N}(0,1)$ entries.

The individual treatment effect (ITE) depends on both observed and latent features:
\begin{align}
\tau(\mathbf{X})
&= 0.5
+ 0.7\,\tanh\!\left(0.7X_{c,5}+0.3X_{c,6}\right)
- 0.5\,\sin\!\left(0.5X_{c,7}\right)
+ 0.4\,\frac{Z_1}{1+|Z_1|}
- 0.3\,\frac{Z_2}{1+|Z_2|} \nonumber\\
&\quad + 0.25\,\tanh\!\left(\frac{X_{c,8}X_{c,9}}{1+|X_{c,9}|}\right)
+ 0.5\,\frac{\mathrm{cat\_contrib}(\mathbf{X}_k,\mathbf{W}_\tau)}{\sqrt{n_{\text{cat}}}},
\end{align}
with $\mathbf{W}_\tau\in\mathbb{R}^{n_{\text{cat}}\times L}$ another fixed random table. This specification yields substantial heterogeneity through nonlinearities and interactions.

Observed outcomes follow the standard potential-outcomes form:
\begin{equation}
Y=\mu_0(\mathbf{X},U)+W\cdot\tau(\mathbf{X})+\varepsilon,
\end{equation}
where the noise is heteroskedastic:
\begin{equation}
\varepsilon\sim\mathcal{N}\!\left(0,\sigma^2(\mathbf{X})\right),
\qquad
\sigma(\mathbf{X})=0.8+0.2\,\frac{|X_{c,0}|}{1+|X_{c,0}|}.
\end{equation}

\begin{remark}[Salient challenges encoded by the design]
This generator simultaneously captures several obstacles encountered in practice: (i) \emph{high dimensionality} (160 raw features, and 280 after one-hot encoding categorical variables), (ii) \emph{population support mismatch} induced by $\sigma_{\text{rct}}>\sigma_{\text{obs}}$, (iii) \emph{unobserved confounding} in the OBS assignment through $U$, (iv) \emph{nonlinear propensity and outcome mechanisms} that stress model flexibility, (v) \emph{heterogeneous treatment effects} $\tau(\mathbf{X})$ varying across the covariate space, and (vi) \emph{mixed data types} via the combination of continuous and categorical covariates.
\end{remark}

\subsection{Joint Mismatch Visualization In Synthetic Dataset}
\label{app:joint_mismatch}
\begin{figure*}[!t]  
    \centering
    \subfloat[Joint Support Mismatch Across Different Overlap Levels]{
        \includegraphics[width=0.75\linewidth]{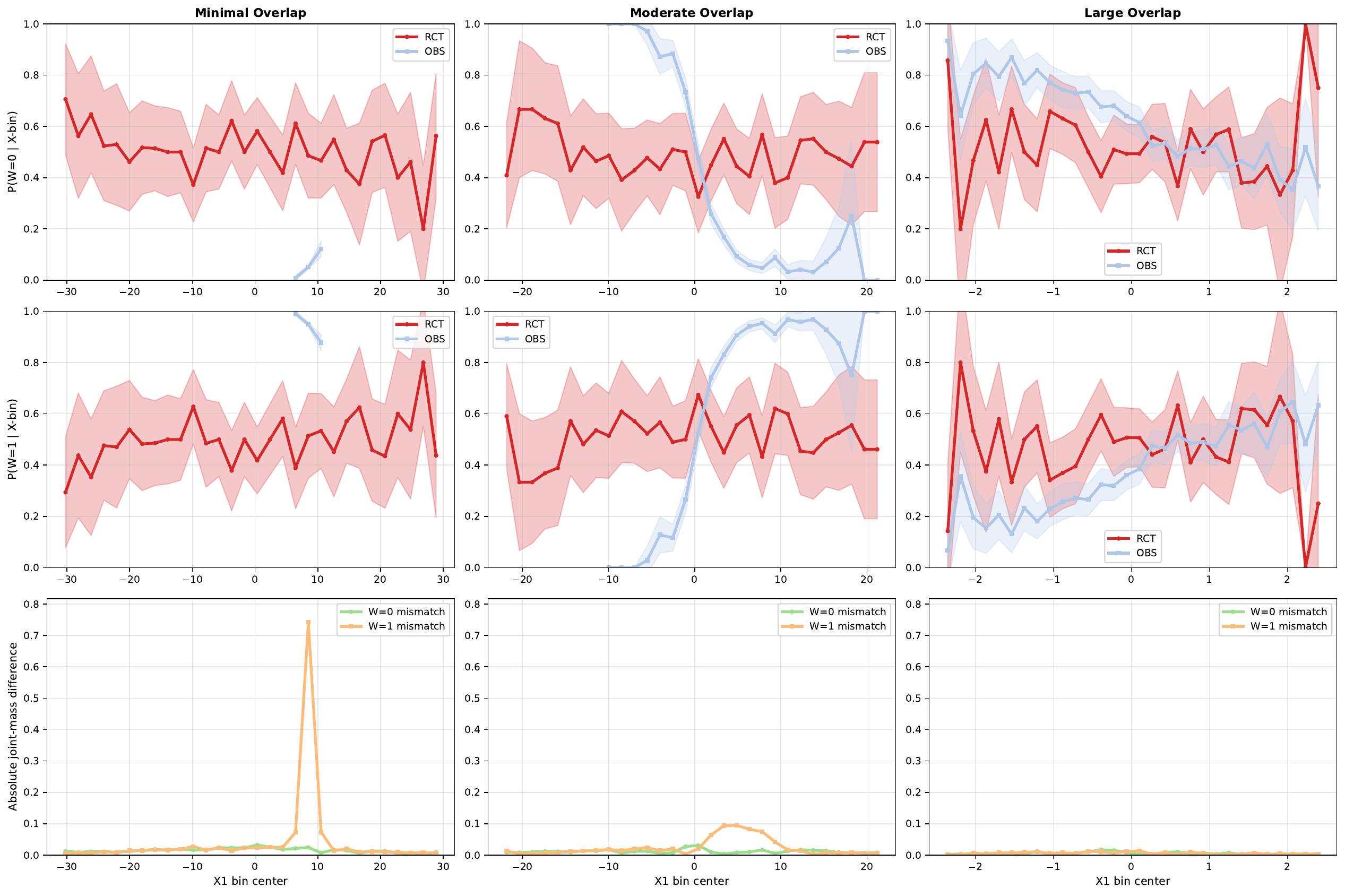}
        \label{fig:single_non-overlap}
    }
    \hfill
    \subfloat[Marginal Support Mismatch Across Different Overlap Levels]{
        \includegraphics[width=0.75\linewidth]{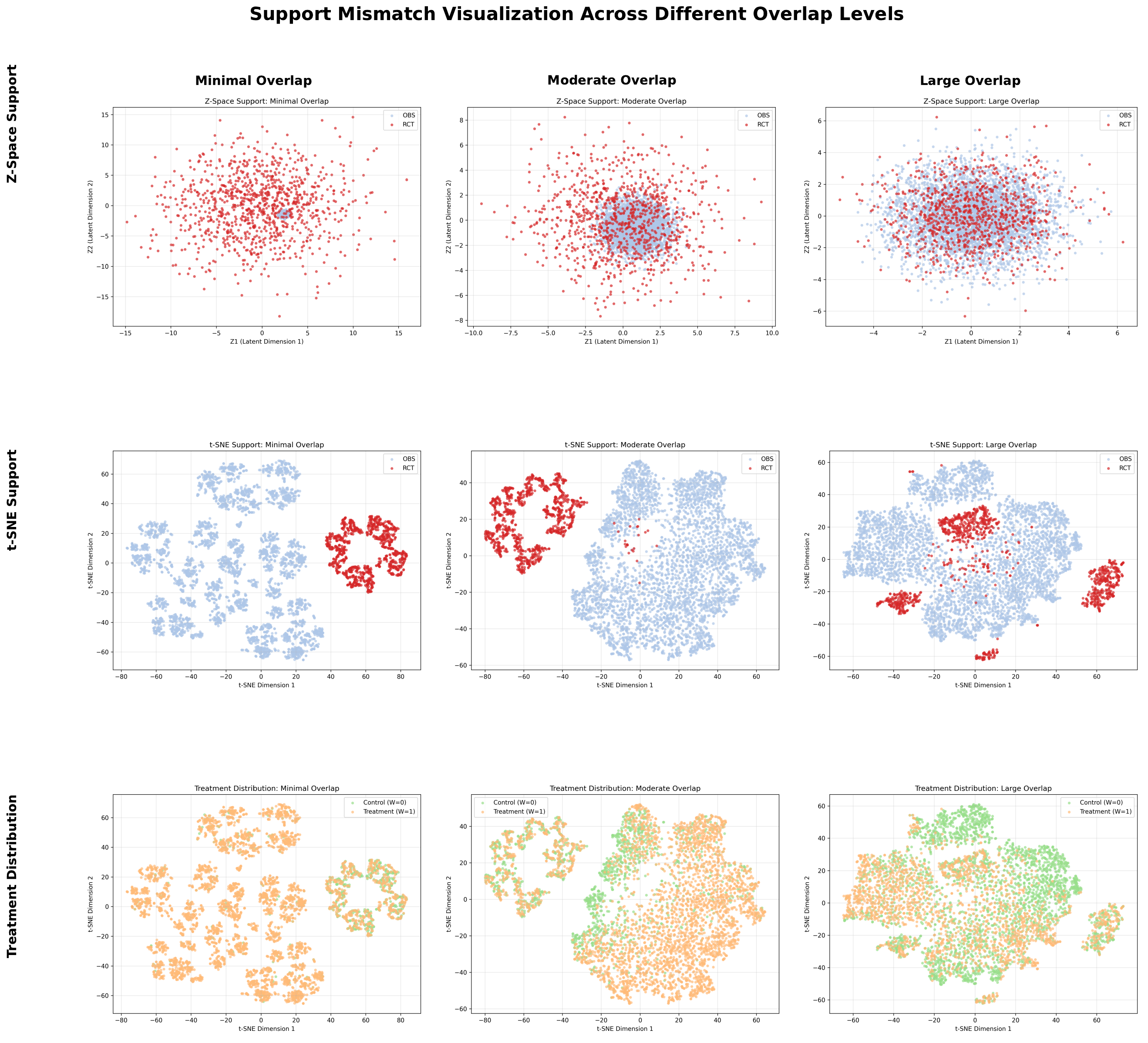}
        \label{fig:joint_non-overlap}
    }
    \caption{Support Mismatch Across Different Overlap Levels (Two Scenarios)}
    \label{fig:all_non-overlap}
\end{figure*}

In the first two rows of Figure~\ref{fig:single_non-overlap}, we explicitly control the support of the OBS and RCT samples through the latent variable $Z$(See Appendix~\ref{app:dgp}), generating datasets with increasing degrees of non-overlap. While the covariate support differs across data sources, the treatment assignments in both RCT and OBS remain well-populated, ensuring that marginal treatment non-overlap (Definition~\ref{def:marginal_nonoverlap}) does not arise. Figure~\ref{fig:joint_non-overlap} instead characterizes \emph{joint} distribution mismatch by comparing the empirical conditional treatment probabilities $P(W=w\mid X\text{-bin})$ for $w\in\{0,1\}$ between RCT and OBS, and quantifying their absolute mass differences across covariate bins. As the latent representations of the two data sources become increasingly aligned, the overlap of their covariate support expands ,from a joint perspective, the discrepancy between conditional treatment probabilities given $X$ diminishes.

\subsection{Baseline Introduction}
\label{app:baseline_intro}
\paragraph{Structurally grounded fusion methods.}
Early approaches to RCT--OBS integration treat the RCT as the primary source of identification and incorporate observational data through explicit structural assumptions on confounding bias.
\begin{itemize}
    \item \textbf{Experimental Grounding.}  
    Corrects observational estimates using RCTs as an external grounding signal, typically assuming that confounding bias admits a low-complexity parametric representation.
    While such methods can achieve consistency under limited overlap and correct specification, they rely on strong structural assumptions and become fragile under pronounced covariate shift, where extrapolation is unavoidable.
\end{itemize}

\paragraph{Statistical integrative estimators with explicit bias decomposition.}
Subsequent work formalizes data fusion by explicitly decomposing observed effects into a true heterogeneous treatment effect (HTE) and a confounding component, establishing identifiability and efficiency gains under transportability and parametric modeling assumptions.
\begin{itemize}
    \item \textbf{Integrative R-Learner.}  
    Extends the R-learner framework to jointly estimate HTEs and confounding functions using both RCT and OBS data.
    This approach improves efficiency over RCT-only estimation but remains fundamentally limited to regions supported by the RCT and does not address structural or conditional non-overlap.
    
    \item \textbf{Yang2022 Elastic Integrative Estimator.}  
    Introduces adaptive or elastic combinations of RCT and OBS information to guard against severe observational bias.
    While improving robustness to misspecification, it focuses on low-dimensional, pre-specified effect modifiers and may discard substantial observational information under strong covariate shift.
    
    \item \textbf{Yang2024 Integrative HTE Estimator.}  
    Further refines integrative estimation through improved modeling of confounding and efficiency trade-offs.
    Despite empirical gains, the method continues to rely on overlap assumptions and does not fundamentally resolve conditional support mismatch.
\end{itemize}

\paragraph{Orthogonal and RCT-centric estimators.}
\begin{itemize}
    \item \textbf{RCT-only R-Learner.}  
    A Neyman-orthogonal R-learner estimated exclusively on RCT data.
    This approach guarantees causal identification but ignores observational information, resulting in high variance and limited efficiency when randomized samples are scarce.
\end{itemize}

\paragraph{Non-structural neural estimators applied to mixed RCT--OBS data.}
More recent practice often applies flexible neural network--based CATE estimators directly to combined RCT and OBS data, relying on representation learning or balancing regularization to mitigate confounding.
\begin{itemize}
    \item \textbf{T-Learner.}  
    Estimates potential outcomes separately for each treatment using flexible predictors.
    Although expressive, it is highly sensitive to covariate imbalance and lacks mechanisms to ensure valid data fusion.
    
    \item \textbf{TARNet.}  
    Learns a shared representation with treatment-specific heads to improve generalization across treatment groups.
    This approach implicitly assumes observational ignorability and offers no protection against conditional non-overlap.
    
    \item \textbf{CFRNet.}  
    Augments TARNet with an IPM-based regularizer to align covariate distributions across treatment groups in the latent space.
    While effective for marginal imbalance, such alignment is insufficient when joint support mismatch arises conditionally on covariates.
    
    \item \textbf{DragonNet.}  
    Further extends TARNet-style architectures with auxiliary propensity prediction to stabilize training.
    Despite empirical improvements in some settings, it remains vulnerable to structural non-overlap and lacks identification-based guarantees for RCT--OBS integration.
\end{itemize}

\paragraph{Ablations of the proposed method.}
To isolate the contribution of each component, we additionally report two ablated variants of the proposed approach:
\begin{itemize}
    \item \textbf{Dual Loss Only.}  
    Enforces moment-based feasibility without explicit overlap-aware alignment, limiting robustness under severe conditional non-overlap.
    
    \item \textbf{Wasserstein Only.}  
    Performs distribution alignment without enforcing feasibility constraints, leading to inadequate identification when support mismatch is structural.
\end{itemize}

Overall, these baselines illustrate the progression from structurally grounded fusion methods to flexible but largely non-structural neural estimators.

\section{Main Results}

\begin{table}[!h]
		\captionsetup{skip=0pt, font=scriptsize}
		\centering
		\caption{Performance comparison under different overlap conditions.}
        \label{tab:pc}
		\resizebox{\textwidth}{!}{
			\begin{tabular}{lcccccc}
				\toprule
				\multirow{2}{*}{\textbf{Method}} & 
				\multicolumn{2}{c}{\textbf{large\_overlap}} & 
				\multicolumn{2}{c}{\textbf{moderate\_overlap}} & 
				\multicolumn{2}{c}{\textbf{minimal\_overlap}} \\
				\cmidrule{2-3} \cmidrule{4-5} \cmidrule{6-7}
				& \textbf{Qini (Mean±Std)} & \textbf{MSE (Mean±Std)} & 
				\textbf{Qini (Mean±Std)} & \textbf{MSE (Mean±Std)} & 
				\textbf{Qini (Mean±Std)} & \textbf{MSE (Mean±Std)} \\
				\midrule
				
				\rowcolor{gray!20}
				\textbf{Own Method (Dual Loss + Wasserstein)} & 
				0.5284±0.1336 & 0.9474±0.1551 & 
				0.5499±0.1202 & 1.0221±0.1116 & 
				0.5137±0.0744 & 0.7961±0.0945 \\
				
				\textbf{Own Method (Dual Loss Only)} & 
				0.1670±0.2077 & 0.6906±0.2788 & 
				0.0942±0.2773 & 0.7517±0.3043 & 
				0.0254±0.1076 & 0.6295±0.2144 \\
				
				\textbf{Own Method (Wasserstein Only)} & 
				0.1224±0.2217 & 0.7433±0.3030 & 
				0.1290±0.2336 & 0.7554±0.2517 & 
				0.0057±0.1450 & 0.5944±0.1601 \\
				
				\textbf{T-Learner} & 
				0.4919±0.1815 & 0.8417±0.3851 & 
				0.5227±0.1810 & 0.7393±0.3602 & 
				0.0951±0.1818 & 0.9844±0.4109 \\
				
				\textbf{TARNet} & 
				0.0370±0.4010 & 1.1621±0.5693 & 
				0.1620±0.3294 & 0.8011±0.2787 & 
				0.0196±0.1728 & 0.8409±0.3920 \\
				
				\textbf{CFRNet} & 
				0.0688±0.3555 & 1.0916±0.6351 & 
				0.1783±0.3664 & 0.8858±0.4217 & 
				-0.0229±0.1741 & 0.7580±0.2997 \\
				
				\textbf{DragonNet} & 
				0.1194±0.4018 & 0.9785±0.5940 & 
				0.1701±0.3106 & 1.0110±0.6626 & 
				0.0379±0.1966 & 0.7777±0.3210 \\
				
				\textbf{Experimental Grounding} & 
				0.2697±0.1900 & 6.1155±1.6609 & 
				0.2282±0.1596 & 7.9774±2.2844 & 
				0.1386±0.1009 & 18.5910±11.5983 \\
				
				\textbf{Integrative R-Learner} & 
				0.3104±0.1916 & 2.0738±0.5419 & 
				0.3037±0.1703 & 2.7009±0.5453 & 
				0.2478±0.0910 & 4.1239±1.6499 \\
				
				\textbf{RCT-only R-Learner} & 
				0.2516±0.2027 & 3.1449±0.7694 & 
				0.2719±0.1756 & 3.6738±0.7417 & 
				0.2094±0.0935 & 5.8326±2.8207 \\
				
				\textbf{Yang2024 Integrative HTE} & 
				0.2762±0.2727 & 0.8671±0.3350 & 
				0.1821±0.3125 & 0.8716±0.3333 & 
				0.1554±0.3307 & 0.7317±0.2409 \\
				
				\textbf{Yang2022 Elastic Integrative} & 
				0.4342±0.2596 & 0.6679±0.2892 & 
				0.1594±0.3327 & 0.9108±0.4571 & 
				0.3414±0.1716 & 0.4254±0.0875 \\
				
				\bottomrule
			\end{tabular}
		}
        
	\end{table}

\begin{table}[!h]
    \captionsetup{skip=0pt, font=scriptsize, justification=centering}
    \centering
    \caption{Performance comparison across datasets.}
    \label{tab:realdata_results}
    \resizebox{\textwidth}{!}{
        \begin{tabular}{lcccccccc}
            \toprule
            \multirow{2}{*}{\textbf{Method}} &
            \multicolumn{2}{c}{\textbf{Dataset 1}} &
            \multicolumn{2}{c}{\textbf{Dataset 2}} &
            \multicolumn{2}{c}{\textbf{Dataset 3}} &
            \multicolumn{2}{c}{\textbf{Dataset 4}} \\
            \cmidrule{2-3} \cmidrule{4-5} \cmidrule{6-7} \cmidrule{8-9}
            & \textbf{QINI} & \textbf{MAPE}
            & \textbf{QINI} & \textbf{MAPE}
            & \textbf{QINI} & \textbf{MAPE}
            & \textbf{QINI} & \textbf{MAPE} \\
            \midrule

            \textbf{Baseline} &
            0.64079 & 0.150023 &  
            0.645225 & 0.145629 &
            0.632018 & 0.094185 &
            0.654108 & 0.09166 \\

            \textbf{OBS-Only} &
            0.606678 & 0.19851 &
            0.619615 & 0.186076 &
            0.605966 & 0.175343 &
            0.639509 & 0.158317 \\

            \textbf{RCT-Only} &
            0.643411 & 0.087209 &
            0.65172 & 0.051301 &
            0.653134 & 0.061832 &
            0.656666 & 0.140676 \\

            \rowcolor{gray!20}
            \textbf{Own-method} &
            0.650734 & 0.07772 &
            0.654392 & 0.028834 &
            0.660753 & 0.061657 &
            0.661447 & 0.048266 \\

            \textbf{Own-method (Joint-only)} &
            0.648422 & 0.09018 &
            0.647558 & 0.044938 &
            0.605987 & 0.114192 &
            0.65322 & 0.062069 \\

            \textbf{Own-method (IPM-only)} &
            0.636355 & 0.066392 &
            0.651405 & 0.097443 &
            0.649587 & 0.06438 &
            0.654507 & 0.060238 \\

            \bottomrule
        \end{tabular}
    }
    \scriptsize \textit{Note}: All data has been desensitized.
\end{table}

\begin{figure*}[!h]
    \centering
    \includegraphics[width=1\linewidth]{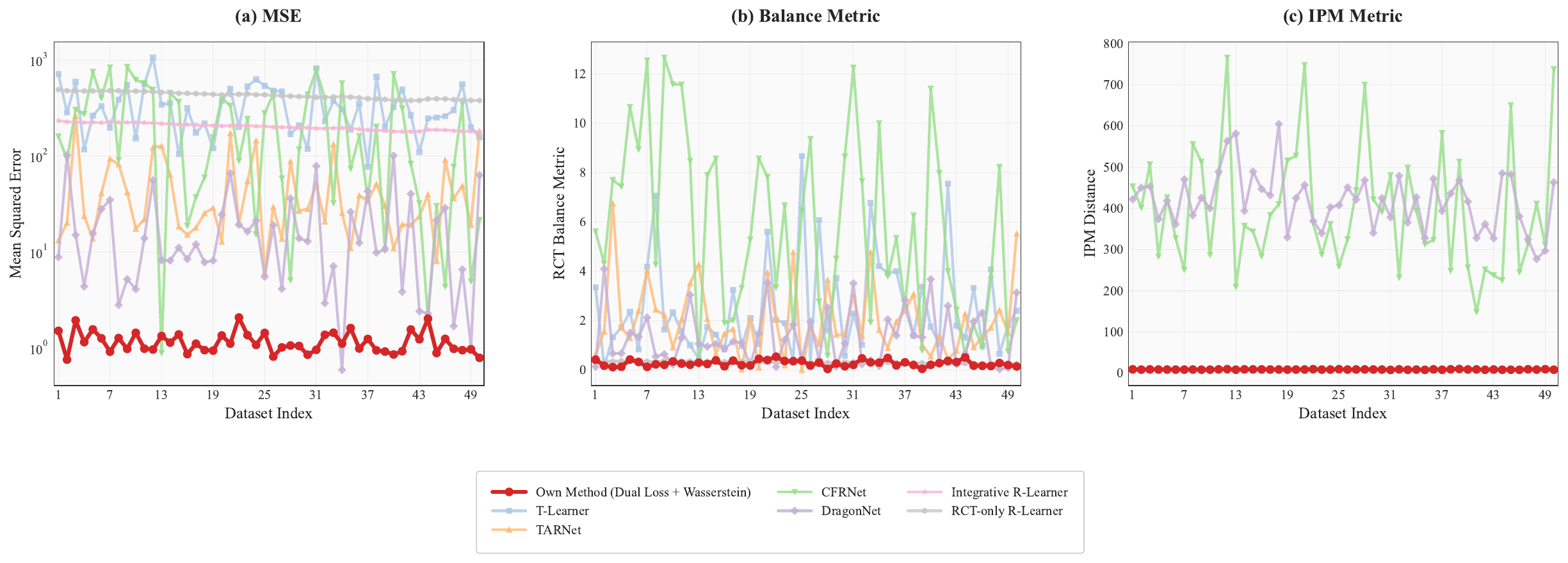}
    \caption{Robustness under Severe Conditional Non-overlap}
    \label{fig:eval_server}
\end{figure*}

\end{document}